\documentclass[a4paper,11pt]{article}
\usepackage[T1]{fontenc} 
\usepackage{float} 

\makeatletter
\gdef\@fpheader{ }
\gdef\@journal{ }
\makeatother
\RequirePackage{amsmath}
\RequirePackage{amssymb}
\RequirePackage{epsfig}
\RequirePackage{graphicx}
\RequirePackage[numbers,sort&compress]{natbib}
\RequirePackage{color}
\RequirePackage[colorlinks=true
,urlcolor=blue
,anchorcolor=blue
,citecolor=blue
,filecolor=blue
,linkcolor=blue
,menucolor=blue
,pagecolor=blue
,linktocpage=true
,pdfproducer=medialab
,pdfa=true
]{hyperref}

\newif\ifnotoc\notocfalse
\newif\ifemailadd\emailaddfalse
\newif\iftoccontinuous\toccontinuousfalse
\makeatletter
\def\@subheader{\@empty}
\def\@keywords{\@empty}
\def\@abstract{\@empty}
\def\@xtum{\@empty}
\def\@dedicated{\@empty}
\def\@arxivnumber{\@empty}
\def\@collaboration{\@empty}
\def\@collaborationImg{\@empty}
\def\@proceeding{\@empty}
\def\@preprint{\@empty}

\newcommand{\subheader}[1]{\gdef\@subheader{#1}}
\newcommand{\keywords}[1]{\if!\@keywords!\gdef\@keywords{#1}\else%
\PackageWarningNoLine{\jname}{Keywords already defined.\MessageBreak Ignoring last definition.}\fi}
\renewcommand{\abstract}[1]{\gdef\@abstract{#1}}
\newcommand{\dedicated}[1]{\gdef\@dedicated{#1}}
\newcommand{\arxivnumber}[1]{\gdef\@arxivnumber{#1}}
\newcommand{\proceeding}[1]{\gdef\@proceeding{#1}}
\newcommand{\xtumfont}[1]{\textsc{#1}}
\newcommand{\correctionref}[3]{\gdef\@xtum{\xtumfont{#1} \href{#2}{#3}}}
\newcommand\jname{JHEP}
\newcommand\acknowledgments{\section*{Acknowledgments}}

\newcommand\preprint[1]{\gdef\@preprint{\hfill #1}}

\newtheorem{theorem}{Theorem}

\newtheorem{corollary}[theorem]{Corollary}

\newenvironment{proof}[1][Proof]{\noindent\textbf{#1.} }{\ \rule{0.5em}{0.5em}}

\makeatother

\newcommand\note[2][]{%
\if!#1!%
\stepcounter{footnote}\footnotetext{#2}%
\else%
{\renewcommand\thefootnote{#1}%
\footnotetext{#2}}%
\fi}

\makeatletter
\newtoks\auth@toks
\renewcommand{\author}[2][]{%
  \if!#1!%
    \auth@toks=\expandafter{\the\auth@toks#2\ }%
  \else
    \auth@toks=\expandafter{\the\auth@toks#2$^{#1}$\ }%
  \fi
}
\makeatother
\makeatletter
\newtoks\affil@toks\newif\ifaffil\affilfalse
\newcommand{\affiliation}[2][]{%
\affiltrue
  \if!#1!%
    \affil@toks=\expandafter{\the\affil@toks{\item[]#2}}%
  \else
    \affil@toks=\expandafter{\the\affil@toks{\item[$^{#1}$]#2}}%
  \fi
}
\makeatother
\makeatletter
\newtoks\email@toks\newcounter{email@counter}%
\newcommand{\emailAdd}[1]{%
\emailaddtrue%
\ifnum\theemail@counter>0\email@toks=\expandafter{\the\email@toks, \@email{#1}}%
\else\email@toks=\expandafter{\the\email@toks\@email{#1}}%
\fi\stepcounter{email@counter}}
\newcommand{\@email}[1]{\href{mailto:#1}{\tt #1}}
\makeatother

\makeatletter
\newcommand*\collaboration[1]{\gdef\@collaboration{#1}}
\newcommand*\collaborationImg[2][]{\gdef\@collaborationImg{#2}}
\makeatletter
\newcommand\afterLogoSpace{\smallskip}
\newcommand\afterSubheaderSpace{\vskip3pt plus 2pt minus 1pt}
\newcommand\afterProceedingsSpace{\vskip21pt plus0.4fil minus15pt}
\newcommand\afterTitleSpace{\vskip23pt plus0.06fil minus13pt}
\newcommand\afterRuleSpace{\vskip23pt plus0.06fil minus13pt}
\newcommand\afterCollaborationSpace{\vskip3pt plus 2pt minus 1pt}
\newcommand\afterCollaborationImgSpace{\vskip3pt plus 2pt minus 1pt}
\newcommand\afterAuthorSpace{\vskip5pt plus4pt minus4pt}
\newcommand\afterAffiliationSpace{\vskip3pt plus3pt}
\newcommand\afterEmailSpace{\vskip16pt plus9pt minus10pt\filbreak}
\newcommand\afterXtumSpace{\par\bigskip}
\newcommand\afterAbstractSpace{\vskip16pt plus9pt minus13pt}
\newcommand\afterKeywordsSpace{\vskip16pt plus9pt minus13pt}
\newcommand\afterArxivSpace{\vskip3pt plus0.01fil minus10pt}
\newcommand\afterDedicatedSpace{\vskip0pt plus0.01fil}
\newcommand\afterTocSpace{\bigskip\medskip}
\newcommand\afterTocRuleSpace{\bigskip\bigskip}
\newlength{\affiliationsSep}
\newcommand\beforetochook{\pagestyle{myplain}\pagenumbering{roman}}

\DeclareFixedFont\trfont{OT1}{phv}{b}{sc}{11}

\renewcommand\maketitle{
\pagestyle{empty}
\thispagestyle{titlepage}
\setcounter{page}{0}
\noindent{\small\scshape\@fpheader}\@preprint\par

\afterLogoSpace
\if!\@subheader!\else\noindent{\trfont{\@subheader}}\fi
\afterSubheaderSpace
\if!\@proceeding!\else\noindent{\sc\@proceeding}\fi
\afterProceedingsSpace
{\LARGE\flushleft\sffamily\bfseries\@title\par}
\afterTitleSpace
\hrule height 1.5\p@%
\afterRuleSpace
\if!\@collaboration!\else
{\Large\bfseries\sffamily\raggedright\@collaboration}\par
\afterCollaborationSpace
\fi
\if!\@collaborationImg!\else
{\normalsize\bfseries\sffamily\raggedright\@collaborationImg}\par
\afterCollaborationImgSpace
\fi
{\bfseries\raggedright\sffamily\the\auth@toks\par}
\afterAuthorSpace
\ifaffil\begin{list}{}{%
\setlength{\leftmargin}{0.28cm}%
\setlength{\labelsep}{0pt}%
\setlength{\itemsep}{\affiliationsSep}%
\setlength{\topsep}{-\parskip}}
\itshape\small%
\the\affil@toks
\end{list}\fi
\afterAffiliationSpace
\ifemailadd 
\noindent\hspace{0.28cm}\begin{minipage}[l]{.9\textwidth}
\begin{flushleft}
\textit{E-mail:} \the\email@toks
\end{flushleft}
\end{minipage}
\else 
\PackageWarningNoLine{\jname}{E-mails are missing.\MessageBreak Plese use \protect\emailAdd\space macro to provide e-mails.}
\fi
\afterEmailSpace
\if!\@xtum!\else\noindent{\@xtum}\afterXtumSpace\fi
\if!\@abstract!\else\noindent{\renewcommand\baselinestretch{.9}\textsc{Abstract:}}\ \@abstract\afterAbstractSpace\fi
\if!\@keywords!\else\noindent{\textsc{Keywords:}} \@keywords\afterKeywordsSpace\fi
\if!\@arxivnumber!\else\noindent{\textsc{ArXiv ePrint:}} \href{http://arxiv.org/abs/\@arxivnumber}{\@arxivnumber}\afterArxivSpace\fi
\if!\@dedicated!\else\vbox{\small\it\raggedleft\@dedicated}\afterDedicatedSpace\fi
\ifnotoc\else
\iftoccontinuous\else\newpage\fi
\beforetochook\hrule
\tableofcontents
\afterTocSpace
\hrule
\afterTocRuleSpace
\fi
\setcounter{footnote}{0}
\pagestyle{myplain}\pagenumbering{arabic}
} 

\renewcommand{\baselinestretch}{1.1}\normalsize
\divide\@tempdima\baselineskip
\@tempcnta=\@tempdima
\skip\@mpfootins = \skip\footins
\renewcommand{\@dotsep}{10000}

\newcommand\ps@myplain{
\pagenumbering{arabic}
\renewcommand\@oddfoot{\hfill-- \thepage\ --\hfill}
\renewcommand\@oddhead{}}
\let\ps@plain=\ps@myplain

\newcommand\ps@titlepage{\renewcommand\@oddfoot{}\renewcommand\@oddhead{}}

\numberwithin{equation}{section}

\renewcommand\section{\@startsection{section}{1}{\z@}%
                                   {-3.5ex \@plus -1.3ex \@minus -.7ex}%
                                   {2.3ex \@plus.4ex \@minus .4ex}%
                                   {\normalfont\large\bfseries}}
\renewcommand\subsection{\@startsection{subsection}{2}{\z@}%
                                   {-2.3ex\@plus -1ex \@minus -.5ex}%
                                   {1.2ex \@plus .3ex \@minus .3ex}%
                                   {\normalfont\normalsize\bfseries}}
\renewcommand\subsubsection{\@startsection{subsubsection}{3}{\z@}%
                                   {-2.3ex\@plus -1ex \@minus -.5ex}%
                                   {1ex \@plus .2ex \@minus .2ex}%
                                   {\normalfont\normalsize\bfseries}}
\renewcommand\paragraph{\@startsection{paragraph}{4}{\z@}%
                                   {1.75ex \@plus1ex \@minus.2ex}%
                                   {-1em}%
                                   {\normalfont\normalsize\bfseries}}
\renewcommand\subparagraph{\@startsection{subparagraph}{5}{\parindent}%
                                   {1.75ex \@plus1ex \@minus .2ex}%
                                   {-1em}%
                                   {\normalfont\normalsize\bfseries}}

\def\fnum@figure{\textbf{\figurename\nobreakspace\thefigure}}
\def\fnum@table{\textbf{\tablename\nobreakspace\thetable}}

\long\def\@makecaption#1#2{%
  \vskip\abovecaptionskip
  \sbox\@tempboxa{\small #1. #2}%
  \ifdim \wd\@tempboxa >\hsize
    \small #1. #2\par
  \else
    \global \@minipagefalse
    \hb@xt@\hsize{\hfil\box\@tempboxa\hfil}%
  \fi
  \vskip\belowcaptionskip}

\renewenvironment{thebibliography}[1]{%
\begin{oldthebibliography}{#1}%
\small%
\raggedright%
\setlength{\itemsep}{5pt plus 0.2ex minus 0.05ex}%
}%
{%
\end{oldthebibliography}%
}

\begin{document}


\title{\boldmath Unified framework for generalized quantum statistics: canonical partition function,
 maximum occupation number, and permutation phase of wave function}


\author[a,b,1]{Chi-Chun Zhou}\note{zhouchichun@dali.edu.cn}

\author[b,2]{and Wu-Sheng Dai}\note{daiwusheng@tju.edu.cn.}


\affiliation[a]{School of Engineering, Dali University, Dali, Yunnan 671003, PR China}
\affiliation[b]{Department of Physics, Tianjin University, Tianjin 300350, PR China}










\abstract{Beyond Bose and Fermi statistics, there still exist various kinds of generalized quantum statistics. 
Two ways to approach generalized quantum statistics: (1) in quantum mechanics, generalize the permutation 
symmetry of the wave function and (2) in statistical mechanics, generalize the maximum occupation 
number of quantum statistics. The connection between these two approaches, however, is obscure. 
In this paper, we suggest a unified framework to describe various kinds of generalized quantum statistics. 
We first provide a general formula of canonical partition functions of ideal $N$-particle gases obeying various kinds 
of generalized quantum statistics. Then we reveal the connection between the permutation phase of the wave 
function and the maximum occupation number, through constructing a method to obtain the permutation 
phase and the maximum occupation number from the canonical partition function. We show that the 
permutation phase of the wave function is closely related to the higher dimensional representation 
of the permutation group. In our scheme, for generalized quantum statistics, the permutation phase of wave 
functions is generalized to a matrix phase, rather than a number. The permutation phase of Bose or 
Fermi wave function, $e^{i0}=1$ or $e^{i\pi}=-1$, is regarded as $1\times1$ matrices, as special cases of generalized 
statistics. It is commonly accepted that different kinds of statistics are distinguished by the maximum 
number. We show that the maximum occupation number is not sufficient to distinguish different kinds 
of generalized quantum statistics. As examples, we discuss a series of generalized quantum statistics in the unified 
framework, giving the corresponding canonical partition functions, maximum occupation numbers, and the 
permutation phase of wave functions. Especially, we propose three new kinds of generalized quantum statistics 
which seem to be the missing pieces in the puzzle. The mathematical basis of the scheme are the 
mathematical theory of the invariant matrix, the Schur-Weyl duality, the symmetric function, and 
the representation theory of the permutation group and the unitary group. The result in this paper 
builds a bridge between the statistical mechanics and such mathematical theories.
}

\maketitle
\flushbottom


\section{Introduction}
The principle of indistinguishability requires that exchanging two identical 
particles does not lead to any observable effect
\cite{reichl2009modern,pathria2011statistical,leinaas1977theory}. In quantum mechanics, 
a physical system is described by a complex wave
function, but the observable is a real number. The principle of
indistinguishability allows a change on wave functions after exchanging two
identical particles so long as the observable does not change. Consequently,
the wave function may change a phase factor after exchanging identical
particles. It comes naturally Bose-Einstein statistics and Fermi-Dirac
statistics whose phase factors change $e^{i0}=1$ and $e^{i\pi}=-1$, 
respectively. Beside
Bose-Einstein statistics and Fermi-Dirac statistics, however, one can still
consider other kinds of quantum statistics so long as it does not violate the
principle of indistinguishability, i.e., there are no changes on observables
after exchanging identical particles. Generalizing quantum statistics along
this line is to consider phase factors between $0$ and $\pi$, e.g., anyons are
successful in explaining the fractional quantum Hall effect
\cite{khare2005fractional,haldane1991fractional}.

In statistical mechanics, macroscopic systems are treated averagely. The
number of microstates is the key in the calculation of the average value.
Particles are indistinguishable, so exchanging particles occupying different
states does not lead to new microstates. In quantum statistics allowed by the
principle of indistinguishability, there is no new microstate after exchanging
identical particles. From the view of statistical mechanics, the difference
between various quantum statistics is reflected in the maximum occupation
numbers, e.g., for Fermi-Dirac statistics the maximum occupation number is $1$
and for Bose-Einstein statistics there is no limitation on the maximum
occupation number. Generalizing statistics along this line is to generalize the
maximum occupation number to arbitrary number, e.g., the spin wave satisfies
Gentile statistics \cite{dai2009intermediate}.

The above analyses shows two approaches of constructing generalized
statistics: (1) in quantum mechanics, generalize the permutation symmetry of the wave function
and (2) in statistical mechanics, generalize the maximum occupation number. The 
connection between those two approaches, however, is obscure.

In statistical mechanics, a system with the fixed particle number should be considered 
in the canonical ensemble and the canonical partition function is the key. 
It is because all the thermodynamic information 
is embedded in the canonical partition function. 
For example, the eigenvalue spectrum can be obtained from the canonical partition 
function \cite{zhou2018calculating}. However, to calculate the canonical partition function
is difficult. It is because one has to deal with the inter-particle interactions 
and at the same time takes the
constraint of fixed particle number into consideration. For example, the previous 
work \cite{zhou2018canonical} gives the canonical partition function for ideal Bose, 
Fermi, and Gentile statistics. 

In this paper, based on the mathematical theory of the invariant matrix \cite{littlewood1977theory}, 
the Schur-Weyl duality \cite{meijer2017schur}, and the symmetric function 
\cite{littlewood1977theory,macdonald1998symmetric}, we suggest a unified framework to 
describe various kinds of generalized quantum statistics. 
We first provide a general formula of canonical partition functions of ideal $N$-particle gases who obeying 
various kinds of generalized quantum statistics. Then we reveal the connection between the permutation 
phase of the wave function and the maximum occupation number, through constructing a method to 
obtain the permutation phase and the maximum occupation number from the canonical partition function. 

We show that the permutation phase of the wave function is closely related to the higher dimensional 
representation of the permutation group. In our scheme, for generalized quantum statistics, the permutation 
phase of wave functions is generalized to a matrix, rather than a number. 
The permutation phase of Bose or Fermi wave function, $1$ or $-1$, is regarded as $1\times1$ matrices, 
as special cases of generalized quantum statistics. 

It is commonly accepted that different kinds of statistics are distinguished by the maximum 
number. We show that the maximum occupation number is not sufficient to distinguish different kinds 
of generalized quantum statistics. 

As examples, we describe various kinds of statistics in a unified framework,
including parastatistics proposed in 1952 by H. Green \cite{green1953generalized,ohnuki1982quantum}, the
intermediate statistics or Gentile statistics proposed in 1940 by G.
Gentile Jr \cite{gentile1940itosservazioni,dai2004gentile}, Gentileonic
statistics proposed by Cattani and Fernandes in
1984 \cite{cattani1984general}, and the immannons proposed by Tichy in 2017 \cite{tichy2017extending},
etc, where the canonical partition function, the maximum occupation number, and the permutation phase 
of the wave function are given. Especially, we propose three new kinds of generalized 
statistics which seem to be the missing pieces in the puzzle.

The mathematical basis of the scheme are the 
mathematical theory of the invariant matrix \cite{littlewood1977theory}, 
the Schur-Weyl duality \cite{meijer2017schur}, the symmetric function \cite{littlewood1977theory,macdonald1998symmetric},
and the representation theory of the permutation group and the 
unitary group \cite{iachello2006lie,hamermesh1962group}. The result in this paper 
builds a bridge between the quantum statistical mechanics and such mathematical theories
and enables one to use the fruitful result in the theory of the symmetric function to solve the
problem in quantum statistical mechanics.

There are discussions on various kinds of generalized quantum statistics. For example, the
generalized quantum statistics such as parastatistics and Gentile statistics are discussed in
Ref. \cite{cattani2009intermediate}. The distinctions between the intermediate statistics, parastatistics, and
Okayama statistics are discussed in Ref. \cite{katsura1970intermediate}.
The connection between the irreducible representation of $S_{N}$ and the
parastatistics is given by Okayama \cite{okayama1952generalization}.
The operator realization of Gentile statistics is given in Ref. \cite{dai2004representation} 
and the statistical distribution 
of various intermediate statistics is calculated from operator relations in
Ref. \cite{dai2012calculating}. 
The relation between properties of Gentile
statistics and the fractional statistics of anyons is discussed in Ref.
\cite{shen2010relation}. The operator realization of intermediate
statistics is discussed in Ref. \cite{shen2007intermediate}.
Statistical distributions for generalized ideal gas of
fractional-statistics gases are given in Ref. \cite{wu1994statistical}.

This paper is organized as follows. In Sec. 2, we briefly review the
mathematical theory of integer partitions and symmetric functions. In Sec.
3, we give a general formula of the canonical partition function of 
an ideal $N$-particle gases under various kinds of statistics. 
In Sec. 4, we give a method to obtain the maximum occupation number 
and show that the maximum occupation number is not sufficient to distinguish different statistics. 
In Sec. 5, we give a method to obtain the permutation phase from the canonical partition 
function and discuss the permutation phase of the wave function for different statistics. 
In Sec. 6, we give a unified framework to describe a series of generalized
statistics. The canonical partition function, the maximum occupation number, 
and the permutation phase of the wave function are given. 
The intermediate statistics such as Gentile
statistics, parastatistics, the immanonns, Gentileonic statistics, and
the new proposed generalized quantum statistics is discussed as examples. In Sec. 7,
the conclusion is given. Some details of the calculation are given
in appendixes.

\section{The integer partition and the symmetric function: a brief review}

The main result of the present paper involves some basic knowledges of the
mathematical theory of the integer partition and the symmetric function. In
this section, we give a brief review on the theory. For more details, one can refer to
Refs.
\cite{littlewood1977theory,vilenkin2013representation,macdonald1998symmetric}%
.

\subsection{The integer partition}

\textit{The integer partition and the length of an integer partition.}\textbf{
}An integer $N$ can be represented as a sum of other integers:
\begin{equation}
N=\lambda_{1}+\lambda_{2}+\ldots+\lambda_{l},
\end{equation}
where $\lambda_{1}\geq\lambda_{2}\geq\ldots\geq\lambda_{l}>0$. The integer
partition of $N$ is denoted by the notation $\left(  \lambda\right)  =\left(
\lambda_{1},\lambda_{2},\ldots,\lambda_{l}\right)  $. The number of the integer in
$\left(  \lambda\right)  $ is the length of $\left(  \lambda\right)  $, denoted
by $l_{\left(  \lambda\right)  }$. $N$ is the size of $\left(  \lambda\right)
$. For example, for an integer partition $\left(  \lambda\right)  =\left(
3,2,1\right)  $, the length is $l_{\left(  \lambda\right)  }=3$ and the size
is $6$.

\textit{The unrestricted partition function }$P(N)$\textit{ and arranging
integer partitions in a prescribed order. }An integer $N$ has many integer
partitions and the unrestricted partition function $P(N)$ counts the number of
integer partitions [3]. For a given $N$, one arranges the integer partition in
the following order: $\left(  \lambda\right)  $, $\left(  \lambda\right)
^{\prime}$, when $\lambda_{1}$ $>$ $\lambda_{1}^{\prime}$; $\left(
\lambda\right)  $, $\left(  \lambda\right)  ^{\prime}$, when $\lambda
_{1}=\lambda_{1}^{\prime}$ but $\lambda_{2}$ $>$ $\lambda_{2}^{\prime}$; and
so on. One keeps comparing $\lambda_{i}$ and $\lambda_{i}^{\prime}$ until all
the integer partitions of $N$ are arranged in a prescribed order. $\left(
\lambda\right)  _{I}$ is the $Ith$ integer partition function. For example,
the integer partitions of $4$ are $\left(  \lambda\right)  _{1}=\left(
4\right)  $, $\left(  \lambda\right)  _{2}=\left(  3,1\right)  $, $\left(
\lambda\right)  _{3}=\left(  2^{2}\right)  $, $\left(  \lambda\right)
_{4}=\left(  2,1^{2}\right)  $, and $\left(  \lambda\right)  _{5}=\left(
1^{4}\right)  $, where, e.g., the superscript in $1^{4}$ means $1$ appearing
twice, the superscript in $2^{2}$ means $2$ appearing twice, and so on.

\textit{The conjugate integer partition.} For an integer partition $\left(
\lambda\right)  =\left(  \lambda_{1},\lambda_{2},\ldots,\lambda_{l}\right)  $,
there is one and only one integer partition $\left(  \lambda\right)  ^{\ast}$
that is conjugate to $\left(  \lambda\right)  $. To get the conjugate integer
partition $\left(  \lambda\right)  ^{\ast}$ from $\left(  \lambda\right)  $,
an efficient method is to use the Young diagram
\cite{littlewood1977theory,macdonald1998symmetric,andrews1998theory}. For
example, the conjugate integer partition of $\left(  \lambda\right)  =\left(
3,1\right)  $, as shown in Fig. (\ref{duality}), is $\left(  \lambda\right)
^{\ast}=\left(  2,1^{2}\right)  $.

\begin{figure}[H]
\centering
\includegraphics[width=0.8\textwidth]{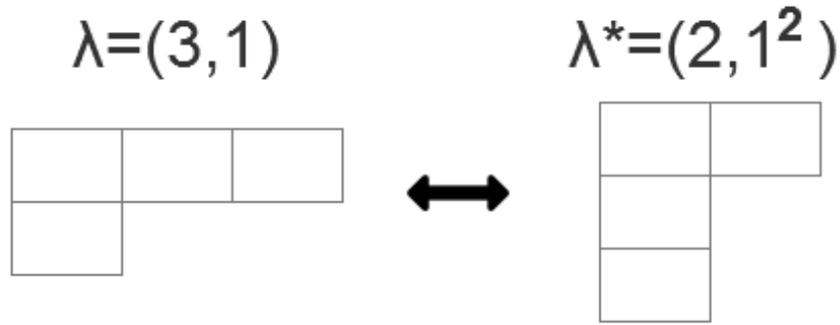}
\caption{An example of the method to obtain the conjugate integer partition by Young diagram.}
\label{duality}
\end{figure}

\subsection{The symmetric function}

The symmetric function, first studied by Hall in 1950s, is an important
issue in algebraic combinatorics
\cite{littlewood1977theory,macdonald1998symmetric}. It is closely related to
the integer partition in number theory and plays an important role in the
theory of group representations
\cite{littlewood1977theory,vilenkin2013representation,macdonald1998symmetric}.

In this section, we give a brief review on several important types of
symmetric functions, such as the S-function $s_{\left(  \lambda\right)
}\left(  x_{1},x_{2},x_{3},\ldots\right)  $, the monomial symmetric function
$m_{\left(  \lambda\right)  }\left(  x_{1},x_{2},x_{3},\ldots\right)  $, and the
power sum symmetric function $p_{\left(  \lambda\right)  }\left(  x_{1}%
,x_{2},x_{3},\ldots\right)  $.

\textit{The S-function }$s_{\left(  \lambda\right)  }\left(  x_{1},x_{2}%
,x_{3},\ldots\right)  $. The S-function, also called the Schur function, is an
important type of the symmetric function. For an integer partition $\left(
\lambda\right)  $ of the integer $N$, the S-function is defined by
\cite{littlewood1977theory,macdonald1998symmetric}\textit{ }%
\begin{equation}
s_{\left(  \lambda\right)  }\left(  x_{1},x_{2},...,x_{n}\right)  =\frac
{\det\left(  x_{i}^{\lambda_{j}+n-j}\right)  }{\det\left(  x_{i}^{n-j}\right)
},
\end{equation}
where $\det (A)$ represents the determinate of matrix $A$. There is also another definition of $s_{\left(  \lambda\right)  }\left(
x_{1},x_{2},...\right)  $ without limitation on the number of variables
\cite{littlewood1977theory,macdonald1998symmetric}:
\begin{equation}
s_{\left(  \lambda\right)  _{I}}\left(  x_{1},x_{2},...\right)  =\sum
_{J=1}^{P\left(  N\right)  }\frac{g_{J}}{N!}\chi_{J}^{I}%
{\displaystyle\prod\limits_{m=1}^{k}}
\left(  \sum_{i}x_{i}^{m}\right)  ^{a_{J,m}}, \label{XS036}%
\end{equation}
where $\left(  \lambda\right)  _{I}$ is the $Ith$ integer partition of the
integer $N$ and $a_{J,m}$ represents the times $m$ occuring in $\left(
\lambda\right)  _{J}$, $\chi_{J}^{I}$ is the simple characteristic of the
permutation group of order $N$. $g_{I}$ satisfies%
\begin{equation}
g_{I}=N!\left(
{\displaystyle\prod\limits_{j=1}^{N}}
j^{a_{I,j}}a_{I,j}!\right)  ^{-1}.
\end{equation}

\textit{The monomial symmetric function }$m_{\left(  \lambda\right)  }\left(
x_{1},x_{2},\ldots\right)  $. For an integer partition $\left(  \lambda
\right)  $, there is a corresponding monomial symmetric polynomial, defined by \cite{littlewood1977theory,macdonald1998symmetric}
\begin{equation}
m_{\left(  \lambda\right)  }\left(  x_{1},x_{2},\ldots,x_{n}\right)
=\sum_{perm}x_{1}^{\lambda_{1}}x_{2}^{\lambda_{2}}\ldots., \label{defm}%
\end{equation}
where $\sum_{perm}$ indicates that the summation runs over all possible
monotonically increasing permutations of $x_{i}$.

\textit{The power sum symmetric function }$p_{\left(  \lambda\right)  }\left(
x_{1},x_{2},\ldots\right)  $. For an integer partition $\left(  \lambda
\right)  $, there is a symmetric polynomial, defined by \cite{littlewood1977theory,macdonald1998symmetric}:
\begin{equation}
p_{\left(  \lambda\right)  }\left(  x_{1},x_{2},\ldots,x_{n}\right)  =%
{\displaystyle\prod\limits_{i=1}^{l_{\left(  \lambda\right)  }}}
m_{\left(  \lambda_{i}\right)  }.
\end{equation}

\textit{The relation among }$s_{\left(  \lambda\right)  }\left(  x_{1}%
,x_{2},...\right)  $\textit{, }$m_{\left(  \lambda\right)  }\left(
x_{1},x_{2},\ldots,x_{n}\right)  $\textit{, and }$p_{\left(  \lambda\right)
}\left(  x_{1},x_{2},\ldots,x_{n}\right)  $. There is a relation between
$m_{\left(  \lambda\right)  }\left(  x_{1},x_{2},\ldots,x_{n}\right)  $ and
$s_{\left(  \lambda\right)  }\left(  x_{1},x_{2},...\right)  $: the S-function
can be represented as a linear combination of the monomial symmetric
polynomial \cite{littlewood1977theory,macdonald1998symmetric}, i.e.,
\begin{equation}
s_{\left(  \lambda\right)  _{K}}\left(  x_{1},x_{2},...\right)  =\sum
_{I=1}^{P\left(  N\right)  }k_{K}^{I}m_{\left(  \lambda\right)  _{I}}\left(
x_{1},x_{2},...\right)  , \label{smgx}%
\end{equation}
where $k_{K}^{I}$ is the Kostka number
\cite{littlewood1977theory,macdonald1998symmetric}. There is a relation
between $p_{\left(  \lambda\right)  }\left(  x_{1},x_{2},\ldots,x_{n}\right)
$ and $s_{\left(  \lambda\right)  }\left(  x_{1},x_{2},...\right)  $: the
S-function $s_{\left(  \lambda\right)  }\left(  x_{1},x_{2},...\right)  $ can
be represented as a linear combination of the \textit{power sum symmetric
function }$p_{\left(  \lambda\right)  }\left(  x_{1},x_{2},\ldots\right)  $
\cite{goulden1992immanants}:
\begin{equation}
p_{\left(  \lambda\right)  _{K}}\left(  x_{1},x_{2},...\right)  =\sum
_{I=1}^{P\left(  N\right)  }\chi_{K}^{I}s_{\left(  \lambda\right)  _{I}%
}\left(  x_{1},x_{2},...\right)  , \label{spgx}%
\end{equation}
where $\chi_{K}^{I}$ is the simple characteristics of $S_{N}$
\cite{vilenkin2013representation,goulden1992immanants}.

\section{The canonical partition function for various kinds of statistics: a general result}

In statistical mechanics, the canonical partition function carries all the thermodynamic information
and plays a centeral role. In this section, we provide a general formula of 
the canonical parition function of various
kinds of generalized quantum statistics. 
The canonical partition function of Bose, Fermi, and Gentile statistics 
considered in the previous work \cite{zhou2018canonical} is the special case of the formula. 

The result in this section is a bridge between the maximum occupation number 
and the permutation phase of the wave function.

\begin{theorem}
\label{theorem1} (1) The canonical partition function
of an ideal gas consisting of $N$-identical particles is%
\begin{equation}
Z\left(  \beta,N\right)  =\sum_{I}^{P\left(  N\right)  }C_{I}s_{\left(
\lambda\right)  _{I}}\left(  e^{-\beta\varepsilon_{1}},e^{-\beta
\varepsilon_{2}},...\right)  ,\label{5301}%
\end{equation}
where $C_{I}\ $is an nonnegative integer and $s_{\left(  \lambda\right)  _{I}%
}\left(  e^{-\beta\varepsilon_{1}},e^{-\beta\varepsilon_{2}},...\right)  $ is
the S-function with $\varepsilon_{i}$ the eigenvalue of single particles.

(2) The Hilbert subspace describing the system is 
\begin{equation}
D=%
{\displaystyle\bigoplus\limits_{I=1}^{P\left(  N\right)  }}
\left(  V^{I}\right)  ^{\oplus C_{I}},
\end{equation}
where $V^{I}$ carries the inequivalent and irreducible representation of $S_{N}$ 
corresponding to the integer partition $\left(
\lambda\right)  _{I}$.
\end{theorem}

\begin{proof}
The direct sum decomposition of the $N$-particle Hilbert space $V^{\otimes N}$
and the trace of the operator $e^{-\beta H}$ in the subspace of $V^{\otimes
N}$ are the basises of the proof. Before proving Eqs. (\ref{5301})
and (\ref{5302}), we prove that the $N$-particle Hilbert space $V^{\otimes N}$
can be decomposed into subspaces labeled by the integer partition $\left(
\lambda\right)  $ of $N$ and the trace of the operator $e^{-\beta H}$ in the
subspace is the S-function.

\textit{A brief review on the mathematical theory of the Schur-Weyl duality.}
Let $V$ be a linear space of the dimension $m$. Let $g\in GL\left(
V\right)  $ be a linear operator on $V$. The action of $g$ on a vector
$e_{i_{1}}\otimes e_{i_{2}}\otimes\ldots\otimes e_{i_{N}}$ in $V^{\otimes N}$
is%
\begin{equation}
g\left(  e_{i_{1}}\otimes e_{i_{2}}\otimes\ldots\otimes e_{i_{N}}\right)
\equiv g e_{i_{1}}\otimes g e_{i_{2}}\otimes\ldots\otimes g
e_{i_{N}}, \label{pro1}%
\end{equation}
where, $\left\{  e_{1},e_{2},...\right\}  $ is a basis in $V$. For $\sigma\in
S_{N}$, the action of $\sigma$ on the vector $e_{i_{1}}\otimes e_{i_{2}%
}\otimes\ldots\otimes e_{i_{N}}$ is%
\begin{equation}
\sigma\left(  e_{i_{1}}\otimes e_{i_{2}}\otimes\ldots\otimes e_{i_{N}%
}\right)  \equiv e_{\sigma_{i_{1}}}\otimes e_{\sigma_{i_{2}}}\otimes
\ldots\otimes e_{\sigma_{i_{N}}}. \label{pro2}%
\end{equation}
Eqs. (\ref{pro1}) and (\ref{pro2}) imply that the operator $g$ commute with
$\sigma$ on $V^{\otimes N}$. For $m\geq
N$, the space $V^{\otimes N}$ can be decomposed into a direct sum of the
subspace $V^{I}$ \cite{meijer2017schur,dao2018schur}. The number of $V^{I}$ is
$P\left(  N\right)  $ and each integer partition $\left(  \lambda\right)
_{I}$ of $N$ corresponds to a subspace $V^{I}$
\cite{meijer2017schur,dao2018schur}, i.e.,%
\begin{equation}
V^{\otimes N}=%
{\displaystyle\bigoplus\limits_{I=1}^{P\left(  N\right)  }}
V^{I}.
\end{equation}
The subspace $V^{I}$ carries irreducible representations for $S_{N}$ and
$GL\left(  V\right)  $ with the dimension $f_{I}$ and $R_{I}$ respectively,
where%
\begin{equation}
f_{I}=N!%
{\displaystyle\prod\limits_{i=1,i<j}}
\left(  \lambda_{I,i}-\lambda_{I,j}-i+j\right)
{\displaystyle\prod\limits_{i=1}}
\left[  \left(  l_{\left(  \lambda\right)  }+\lambda_{I,i}-i\right)  !\right]
^{-1}, \label{pro20}%
\end{equation}
and%
\begin{equation}
R_{I}=%
{\displaystyle\prod\limits_{i<j}^{m}}
\left(  a_{I,i}-a_{I,j}+j-i\right)  \left(  j-i\right)  ^{-1}%
\end{equation}
with $a_{I,1}=\lambda_{I,1}$, $a_{I,2}=\lambda_{I,2}$, $\ldots$,
$a_{I,l_{\left(  \lambda\right)  }}=\lambda_{I,l_{\left(  \lambda\right)  }}$,
$a_{I,l_{\left(  \lambda\right)  }+1}=0$,$\ldots$, and $a_{I,m}=0$
\cite{meijer2017schur,dao2018schur}. The dimension of $V^{I\text{ }}$is
\cite{meijer2017schur,dao2018schur}
\begin{equation}
\dim\left(  V^{I}\right)  =f_{I}R_{I}.
\end{equation}
For $S_{N}$, the inequivalent and irreducible representation with the
dimension $f_{I}$\ occurs $R_{I}$ times in $V^{I}$
\cite{meijer2017schur,dao2018schur}. For $GL\left(  V\right)  $, the
inequivalent and irreducible representation of the dimension $R_{I}$\ occurs
$f_{I}$ times in $V^{I}$ \cite{meijer2017schur,dao2018schur}. It can be
verified that
\begin{equation}
\dim\left(  V^{\otimes N}\right)  =m^{N}=\sum_{I=1}^{P\left(  N\right)  }%
\dim\left(  V^{I}\right)  =\sum_{I=1}^{P\left(  N\right)  }R_{I}f_{I}.
\end{equation}

\textit{A brief review on the mathematical theory of the invariant matrix.
}For an $m$-dimensional matrix group $G$. Let $A$ be an $m$-dimensional matrix
in $G$. Let $T\left(  A\right)  $ be a function of $A$. $T\left(  A\right)  $
is an invariant matrix \cite{littlewood1977theory} if%
\begin{equation}
T\left(  A\right)  T\left(  B\right)  =T\left(  AB\right)  ,
\end{equation}
where $B$ is also an $m$-dimensional matrix in $G$. The invariant matrix gives
a representation of the group $G$. If $T$ is reducible, then for any $A$ in $G$,
$T\left(  A\right)  $ can be diagonalized in the same way and the matrix in
diagonal is a new invariant matrix of $G$ \cite{littlewood1977theory}. The
$N$ times direct product of $G$, $G^{\otimes N}$, is an invariant matrix
\cite{littlewood1977theory}. The $G^{\otimes N}$ can be decomposed into
$P\left(  N\right)  $ irreducible invariant matrices. An integer partition
$\left(  \lambda\right)  _{I}$ corresponds to an irreducible invariant matrix,
denoted by $T^{I}\left(  G\right)  $. For $A$ in $G$, the trace of
$T^{I}\left(  A\right)  $ is \cite{littlewood1977theory}
\begin{equation}
tr\left[  T^{I}\left(  A\right)  \right]  =s_{\left(  \lambda\right)  _{I}%
}\left(  a_{1},a_{2},\ldots\right)  ,
\end{equation}
where $a_{i}$ is the eigenvalue of $A$.

\textit{The direct sum decomposition of the }$N$\textit{-particle Hilbert
space }$V^{\otimes N}$. By using the Schur-Weyl duality, the Hilbert space of
a $N$-particle system $V^{\otimes N}$ can be decomposed into a direct sum of
subspaces $V^{I}$. The number of the subspace is $P\left(  N\right)  $ and
each integer partition $\left(  \lambda\right)  _{I}$ of $N$ corresponds to a
subspace $V^{I}$. The dimension of the subspace $V^{I}$ is $R_{I}f_{I}$. The
space $V^{I}$ gives $f_{I}$ equivalent and irreducible representations with
the dimension $R_{I}$ for the Hamiltonian and $R_{I}$ equivalent and
irreducible representations with the dimension $f_{I}$ for $S_{N}$.

\textit{The trace of the operator }$e^{-\beta H}$\textit{ in the subspace.
}Let $H$ be the Hamiltonian of a single particle and $V$ be the Hilber space
of a single particle. One can give the matrix expression of the operator
$e^{-\beta H}$ on $V$%
\begin{equation}
e^{-\beta H}=\sum_{i=1}e^{-\beta\varepsilon_{i}}|\phi_{i}\rangle\langle
\phi_{i}|,
\end{equation}
where $|\phi_{i}\rangle$ is the eigenfunction of the Hamiltonian $H$ and
$\varepsilon_{i}$ is the corresponding eigenvalue. $e^{-\beta H_{N}}=\left(
e^{-\beta H}\right)  ^{\otimes N}$ is an operator on $V^{\otimes N}$, where
$H_{N}=%
{\displaystyle\bigoplus\limits_{i=1}^{N}}
H_{i}$ is the Hamiltonian of an $N$-identical-particle gas system. Since $e^{-\beta
H_{N}}$ is an invariant matrix of $e^{-\beta H}$, by using the mathematical
theory of the invariant matrix\textit{,} $e^{-\beta H_{N}}$ can be decomposed
into $P\left(  N\right)  $ irreducible invariant matrices. Each irreducible
invariant matrix, denoted by $D^{I}\left(  e^{-\beta H}\right)  $,
corresponds to an integer partition $\left(  \lambda\right)  _{I}$ of $N$.
The trace of $D^{I}\left(  e^{-\beta H}\right)  $ is%
\begin{equation}
tr\left[  D^{I}\left(  e^{-\beta H}\right)  \right]  =s_{\left(
\lambda\right)  _{I}}\left(  e^{-\beta\varepsilon_{1}},e^{-\beta
\varepsilon_{2}},\ldots\right)  . \label{pro4}%
\end{equation}
By the mathematical theory of the Schur-Weyl duality\textit{,} we recongnize
that Eq. (\ref{pro4}) is the trace of $e^{-\beta H_{N}}$ in the subspace
$V^{\prime I}$ with $V^{\prime I}$ the inequivalent and irreducible
representation corresponding to the integer partition $\left(  \lambda\right)
_{I}$. That is, for a complete basis $\left\vert \Phi\right\rangle $ in
$V^{\prime I}$, one has
\begin{equation}
\sum\left\langle \Phi\right\vert e^{-\beta H_{N}}\left\vert \Phi\right\rangle
=s_{\left(  \lambda\right)  _{I}}\left(  e^{-\beta\varepsilon_{1}}%
,e^{-\beta\varepsilon_{2}},\ldots\right)  .
\end{equation}
According to the Schur-Weyl duality\textit{, }the
inequivalent and irreducible representations $V^{\prime I}$ occurs $f_{I}$
times in $V^{I}$, that is, for a complete basis $\left\vert \Psi\right\rangle $
in $V^{I}$, the equation
\begin{equation}
\sum\left\langle \Psi\right\vert e^{-\beta H_{N}}\left\vert \Psi\right\rangle
=f_{I}s_{\left(  \lambda\right)  _{I}}\left(  e^{-\beta\varepsilon_{1}%
},e^{-\beta\varepsilon_{2}},\ldots\right)  \label{ca1}%
\end{equation}
holds. In Eq. (\ref{ca1}), the coefficient $f_{I}$ can be canceled by
setting
\begin{equation}
\left\vert \Psi^{\prime}\right\rangle =\frac{1}{\sqrt{f_{I}}}\left\vert
\Psi\right\rangle .
\end{equation}
Thus, we make no distinguish between the subspace $V^{\prime I}$ and $V^{I}$
in the rest discussion of the present paper.

\textit{The proof of Eq. (\ref{5301}).} An
identical-particle system is described in a Hilbert subspace $D$. The
space $D$ can be decomposed into subspace $V^{I}$ that carries the equivalent
and irreducible representation of $S_{N}$ \cite{hamermesh1962group}, i.e.,%
\begin{equation}
D=%
{\displaystyle\bigoplus\limits_{I=1}^{P\left(  N\right)  }}
\left(  V^{I}\right)  ^{\oplus C_{I}},
\end{equation}
where $C_{I}$ are nonnegative integers representing the times of $V^{I}$ occuring
in $D$. By the definition of the canonical partition function, $Z\left(
\beta,N\right)  =tr\left[  D\left(  e^{-\beta H}\right)  \right]  $, and Eq.
(\ref{pro4})\textit{,} we give the canonical partition function of an
$N$-identical-particle system:%
\begin{align}
Z\left(  \beta,N\right)   &  =tr\left[  D\left(  e^{-\beta H}\right)  \right]
\nonumber\\
&  =\sum_{I=1}^{P\left(  N\right)  }C_{I}tr\left[  D^{I}\left(  e^{-\beta
H}\right)  \right] \nonumber\\
&  =\sum_{I=1}^{P\left(  N\right)  }C_{I}s_{\left(  \lambda\right)  _{I}%
}\left(  e^{-\beta\varepsilon_{1}},e^{-\beta\varepsilon_{2}},...\right)  ,
\label{pro5}%
\end{align}
where $D\left(  e^{-\beta H}\right)  $ is the representation of $e^{-\beta H}$
on $D$. Therefore, Eq. (\ref{pro5}) proves Eq. (\ref{5301}). For a basis $|\Phi\rangle$ in
$V^{I}$, it gives the equivalent and irreducible representation of $S_{N}$.
Therefore, Eq. (\ref{pro4}) requires that Eqs. (\ref{5302}) and (\ref{pro6}) hold.
\end{proof}

In the proof of Theorem (\ref{theorem1}), we show that the S-function $s_{\left(
\lambda\right)  }\left(  x_{1},x_{2},\ldots\right)  $ in mathematics is closely related to
the Hilbert subspace: if the canonical partition function is written as a linear
combination of the S-function $s_{\left(  \lambda\right)  }\left(  x_{1}%
,x_{2},\ldots\right)  $, the coefficient gives the Hilbert subspace.

\textit{An example of decomposing the Hilbert space }$V^{N}$\textit{.} For the
sake of clarity, we give an example of the decomposition of the $V^{\otimes
N}$. Let the dimension of the Hilbert space of a single particle $V$ be $6$. For
a system consists of $5$ particles, the Hilbert space $V^{\otimes5}$ is a
$7776$-dimensional space. It can be decomposed into $7$ subspaces, as shown in
Table. (\ref{table}) : the subspace $V^{\left(  5\right)  }$ corresponds
to the integer partition $\left(  5\right)  $ and the dimension of $V^{\left(
5\right)  }$ is $252$. It gives a one-dimensional representation of $S_{N}$,
and a $252$-dimensional representation of the Hamiltonian, and so on.%

\begin{table}[H]  
\caption{An example of decomposing the Hilbert space}
 \label{table}
\centering
 \begin{tabular}{llll}  

\hline   

  $\text{subspace }V^{\left(  \lambda\right)}$ & $\dim\left(  V^{\left(\lambda\right)  }\right)$ &$\text{ }f_{\left(  \lambda\right)  })$ & $R_{\left(
\lambda\right)  }$\\  

\hline   

  $V^{\left(  5\right)  }$ & 252 & 1 &252 \\ 

  $V^{\left(  4,1\right)  }$& 2016 & 4 & 504 \\
  $V^{\left(   3,2\right)  }$& 2100 & 5 & 420 \\
  $V^{\left(  3,1^{2}\right)  }$& 2016 & 6 & 336 \\
  $V^{\left(  2^{2},1\right)  }$& 1050 & 5 & 210 \\
  $V^{\left(  2,1^{3}\right)  }$& 336 & 4 & 84 \\
  $V^{\left(  1^{5}\right)  }$& 6 & 1 & 6 \\

\hline 

\end{tabular}

\end{table}

\section{The maximum occupation number for various kinds of statistics}

In statistical mechanics, it is commonly believed that quantum statistics 
is determined by the maximum occupation number. Different maximum occupation 
numbers lead to different kinds of quantum statistics, e.g., 
setting the maximum occupation number to $1$ leads to Fermi statistics.
Therefore, the maximum occupation number is an important issue. 
For example, Ref. \cite{dai2012calculating} shows that a number of
quantization schemes in quantum field theory corresponds to different maximum
occupation numbers in quantum statistical mechanics.
Ref. \cite{dai2004gentile} considers quantum statistics
where the maximum occupation number for different states 
is different.

In this section, (1) we provide a method to obtain the maximum occupation
number from the canonical partition function, and (2) we point out that the maximum
occupation number is not sufficient to distinguish different kinds of
statistics. 

\subsection{Obtaining the maximum occupation number from the canonical
partition function}

In this section,
we show that as long as the canonical partition function is expressed in terms
of $m_{\left(  \lambda\right)  }\left(  x_{1},x_{2},\ldots\right)  $, the
maximum occupation number can be obtained directly.

\begin{theorem}
\label{theorem2} For an ideal gas consisting of $N$-identical particles, if the canonical partition function $Z\left(
\beta,N\right)  $ can be written in terms of the monomial symmetric polynomial
$m_{\left(  \lambda\right)  }\left(  x_{1},x_{2},\ldots\right)  $ with
nonnegative-integers coefficient, i.e.,%
\begin{equation}
Z\left(  \beta,N\right)  =m_{\left(  \lambda\right)  ^{q}}\left(
e^{-\beta\varepsilon_{1}},e^{-\beta\varepsilon_{2}},...\right)  +\sum
_{\lambda_{I,1}\leq q}^{P\left(  N\right)  }M^{I}m_{\left(  \lambda\right)
_{I}}\left(  e^{-\beta\varepsilon_{1}},e^{-\beta\varepsilon_{2}},...\right)
\label{maximum1}%
\end{equation}
with $\left(  \lambda\right)  ^{q}$ denoting the integer partition with
$\lambda_{1}=q$ and $M^{I}$ a nonnegative integer, then the maximum occupation
number is $q$.
\end{theorem}

\begin{proof}
The canonical partition function is
\begin{equation}
Z\left(  \beta,N\right)  =\sum_{E}\omega\left(  E,N\right)  e^{-\beta E},
\label{maximum2}%
\end{equation}
where $\omega\left(  E,N\right)  $ is the number of the microstate in the
macrostate $(N,E)$ \cite{reichl2009modern,pathria2011statistical}. By letting
$x_{i}=e^{-\beta\varepsilon_{i}}$ with $\varepsilon_{i}$ the energy of quantum
state in Eq. (\ref{defm}), the monomial symmetric polynomial becomes
\begin{equation}
m_{\left(  \lambda\right)  }\left(  e^{-\beta\varepsilon_{1}},e^{-\beta
\varepsilon_{2}},\ldots\right)  =\sum_{perm}e^{-\beta\varepsilon_{i_{1}%
}\lambda_{1}}e^{-\beta\varepsilon_{i_{2}}\lambda_{2}}\ldots. \label{maximum3}%
\end{equation}
If the canonical partition function of a quantum system $Z\left(
\beta,N\right)  $ is the monomial symmetric polynomial $m_{\left(
\lambda\right)  }\left(  e^{-\beta\varepsilon_{1}},e^{-\beta\varepsilon_{2}%
},\ldots\right)  $, i.e.,%
\begin{equation}
Z\left(  \beta,N\right)  =m_{\left(  \lambda\right)  }\left(  e^{-\beta
\varepsilon_{1}},e^{-\beta\varepsilon_{2}},\ldots\right)  , \label{maximum30}%
\end{equation}
Then, substituting Eqs. (\ref{maximum2}) and (\ref{maximum3}) into Eq.
(\ref{maximum30}) gives
\begin{equation}
\sum_{E}\omega\left(  E,N\right)  e^{-\beta E}=\sum_{perm}e^{-\beta
\varepsilon_{i_{1}}\lambda_{1}}e^{-\beta\varepsilon_{i_{2}}\lambda_{2}}\ldots,
\label{maximum4}%
\end{equation}
where $\omega\left(  E,N\right)  $ in Eq. (\ref{maximum4})
counts the number of microstate where there are $\lambda_{1}$ particles
occupying a quantum state, $\lambda_{2}$ particles occupying another quantum
states, and so on \cite{zhou2018canonical,zhou2018statistical}. 
In this case, the maximum occupation number is $\lambda_{1}$.

If the canonical partition function is a linear combination of $m_{\left(
\lambda\right)  }\left(  e^{-\beta\varepsilon_{1}},e^{-\beta\varepsilon_{2}%
},\ldots\right)  $, say, Eq. (\ref{maximum1}), then the corresponding
$\omega\left(  E,N\right)  $ satisfies
\begin{equation}
\sum_{E}\omega\left(  E,N\right)  e^{-\beta E}=m_{\left(  \lambda\right)
^{q}}\left(  e^{-\beta\varepsilon_{1}},e^{-\beta\varepsilon_{2}},...\right)
+\sum_{\lambda_{I,1}\leq q}^{P\left(  N\right)  }M^{I}m_{\left(
\lambda\right)  _{I}}\left(  e^{-\beta\varepsilon_{1}},e^{-\beta
\varepsilon_{2}},...\right)  . \label{maximum5}%
\end{equation}
Introducing $\omega_{I}\left(  E,N\right)  $ that satisfies
\begin{equation}
\omega\left(  E,N\right)  =\sum_{I=1}^{P\left(  N\right)  }\omega_{I}\left(
E,N\right)  \label{maximum6}%
\end{equation}
and substituting Eq. (\ref{maximum6}) into Eq. (\ref{maximum5}) give
\begin{align}
\sum_{E}\omega_{\left(  \lambda\right)  ^{q}}\left(  E,N\right)  e^{-\beta E}
&  =m_{\left(  \lambda\right)  ^{q}}\left(  e^{-\beta\varepsilon_{1}%
},e^{-\beta\varepsilon_{2}},...\right), \\
\sum_{\lambda_{I,1}\leq q}^{P\left(  N\right)  }\sum_{E}\omega_{I}\left(
E,N\right)  e^{-\beta E}  &  =\sum_{\lambda_{I,1}\leq q}^{P\left(  N\right)
}M^{I}m_{\left(  \lambda\right)  _{I}}\left(  e^{-\beta\varepsilon_{1}%
},e^{-\beta\varepsilon_{2}},...\right)  .
\end{align}
Since $\omega_{I}\left(  E,N\right)  $ counts the number of microstate where
there are $\lambda_{I,1}$ particles occupying a quantum state, $\lambda
_{I,2}$ particles occupying another quantum state, and so on. Thus, the
maximum occupation number of the system is the largest $\lambda_{J,1}$, for
$M^{J}\neq0$. In Eq. (\ref{maximum5}), the largest $\lambda_{J,1}$ is $q$.
Therefore the maximum occupation number of the system is $q$.
\end{proof}

In the proof of Theorem (\ref{theorem2}), we show that the monomial symmetric polynomials
$m_{\left(  \lambda\right)  }\left(  x_{1},x_{2},\ldots\right)  $ in mathematics is closely
related to the maximum occupation number in physics: if the canonical partition function is written as a linear
combination of the monomial symmetric polynomials
$m_{\left(  \lambda\right)  }\left(  x_{1},x_{2},\ldots\right)  $, 
the maximum element in the integer partition with non-zero coefficient gives the maximum occupation
number.

\textit{Examples: Bose-Einstein, Fermi-Dirac, and Gentile statistics. } As examples, we give a brief
discussion on Bosee-Einstein, Fermi-Dirac, and Gentile statistics as examples. The
canonical partition functions for ideal Bose, Fermi, and Gentile gases are given in the previous
work \cite{zhou2018canonical}:%
\begin{align}
Z_{B}\left( N,\beta \right)& =m_{\left( N\right) }\left( e^{-\beta
\varepsilon _{1}},e^{-\beta \varepsilon _{2}},\ldots \right) +\sum_{\lambda
_{I,1}<N}m_{\left( \lambda \right) }\left( e^{-\beta \varepsilon
_{1}},e^{-\beta \varepsilon _{2}},\ldots \right)  ,\label{maximum8}\\
Z_{F}\left( N,\beta \right)& =m_{\left( 1^{N}\right) }\left( e^{-\beta
\varepsilon _{1}},e^{-\beta \varepsilon _{2}},\ldots \right) ,\label{maximum9}\\
Z_{q}\left( N,\beta \right)& =m_{\left( \lambda \right) ^{q}}\left( e^{-\beta
\varepsilon _{1}},e^{-\beta \varepsilon _{2}},\ldots \right) +\sum_{\lambda
_{I,1}\leq q}m_{\left( \lambda \right) }\left( e^{-\beta \varepsilon
_{1}},e^{-\beta \varepsilon _{2}},\ldots \right) . \label{maximum10}%
\end{align}
From Eqs. (\ref{maximum8}), (\ref{maximum9}), and (\ref{maximum10}), we can
directly obtain the maximum occupation number: for Fermi cases, the maximum occupation number is $1$; for the Gentile cases, the maximum occupation number is $q$;
for Bose cases, there is no limitation on the maximum occupation number.

\subsection{The failure of distinguishing different statistics by the maximum
occupation number}

It is commonly accepted that various statistics are distinguished by the maximum occupation number.
In this section, we show that the maximum occupation number is not
sufficient to distinguish different kinds of statistics. It is the maximum
occupation number together with the coefficient of the monomial symmetric
polynomial $m_{\left(  \lambda\right)  }\left(  e^{-\beta\varepsilon_{1}%
},e^{-\beta\varepsilon_{2}},\ldots\right)  $ in Eqs. (\ref{maximum1}) 
determines the kind of statistics. 
A discussion of para-Fermi statistics and Gentile statistics is given as examples.

\begin{corollary}
\label{coro2} (1) The maximum occupation number is not sufficient to
distinguish different statistics. The nonnegative-integers
coefficient $M^{I}$ in Eq. (\ref{maximum1}) distinguishes statistics with the
same maximum occupation number $q$. 

(2) The microstate where there 
are $\lambda_{I,1}$ particles occupying a
quantum state, $\lambda_{I,2}$ particles occupying another quantum state,
and so on, will be counted $M^{I}$ times in the number of the microstate
$\omega\left(  E,N\right)  $.
\end{corollary}

The proof of Corollary. (\ref{coro2}) is embedded in the proof of Theorem (\ref{theorem2}).
We give examples to illustrate Corollary. (\ref{coro2}).
 
\textit{Example: Para-Fermi and Gentile statistics}. For the sake of clarity, we give
explicit expressions of the canonical partition function of para-Fermi
statistics and Gentile statistics as examples.
The detail of the calculation is given in the following section. For example, the canonical
partition function for para-Fermi statistics with parameter $q=2$ is%
\begin{equation}
Z_{2}^{PF}\left(  \beta,5\right)  =m_{\left(  2^{2},1\right)  }+3m_{\left(
2,1^{3}\right)  }+10m_{\left(  1^{5}\right)  } \label{maximum13}%
\end{equation}
and the canonical partition function for Gentile statistics of maximum
occupation number $q=2$ is
\begin{equation}
Z_{2}^{G}\left(  \beta,5\right)  =m_{\left(  2^{2},1\right)  }+m_{\left(
2,1^{3}\right)  }+m_{\left(  1^{5}\right)  }, \label{maximum14}%
\end{equation}
where we denote $m_{\left(  \lambda\right)  }=m_{\left(  \lambda\right)
}\left(  e^{-\beta\varepsilon_{1}},e^{-\beta\varepsilon_{2}},...\right)  $ for
convenience. Eqs. (\ref{maximum13}) and (\ref{maximum14}) show the difference between the
para-Fermi statistics and Gentile statistics: although the maximum
occupation numbers both are $2$, the weight for microstate corresponding to
the integer partition $\left(  2,1^{3}\right)  $ and $\left(  1^{5}\right)  $
are different. The microstate with less number of particles occupying the same
states has larger weight in para-Fermi statistics. This is because
exchanging two different-state particles that obey para-Fermi statistics
will lead to a new microstate.

\section{The permutation phase of the wave function for various kinds of statistics}

Unlike that in Bose and Fermi statistics, the permutation phase of the wave 
function of some kinds of generalized quantum statistics with a given maximum occupation 
number is unknown. For example, the permutation phase of the wave function corresponding 
to Gentile statistics which 
is determined by a maximum occupation number is obscure.
 
In this section, we show that for generalized quantum statistics 
with the Hamiltonian invariant 
under permutations, the permutation phase
is generalized to a matrix, rather than a number. 
We provide a method to obtain the permutation phase of the wave function. 
The permutation phase of Bose or Fermi wave function, $1$ or $-1$, is 
regarded as $1\times1$ matrices, as special cases in the scheme. 

We also point out that there is generalized quantum statistics with 
the Hamiltonian nonnvariant 
under permutations. 
For such statistics, 
the permutation phase of the wave function can not be constructed. 
We provide a method to distinguish them.

\subsection{The statistics with the Hamiltonian invariant under permutations: the permutation phase}
In this section, for various kinds of statistics with the Hamiltonian invariant under permutations,
we provide a method to obatain the permutation phase from the canonical partition function.
\begin{theorem}
\label{theorem11} If the Hamiltonian is invariant under permutations, i.e., 
$\left[  H,S_{N}\right]=0$, then the wave function after exchanging two particles $\sigma_{ij}|\Phi\rangle$ satisfies
\begin{equation}
\sigma_{ij}|\Phi\rangle=D\left(  \sigma_{ij}\right)  |\Phi\rangle,
\label{5302}%
\end{equation}
where $\sigma_{ij}|\Phi\rangle$ represents exchanging the $ith$ particle and the $jth$ particle 
in the
wave function $|\Phi\rangle$. $D\left(  \sigma_{ij}\right) $ here is the permutation phase,
\begin{equation}
D\left(  \sigma_{ij}\right)  =%
{\displaystyle\bigoplus\limits_{I=1}^{P\left(  N\right)  }}
\left[  D^{I}\left(  \sigma_{ij}\right)  \right]  ^{\oplus C_{I}}, \label{pro6}%
\end{equation}
where $D^{I}$ is the inequivalent and irreducible representation of $S_{N}$ corresponding to the integer
partition $\left(  \lambda\right)_{I}$, and $D^{I}$ occurs $C_{I}$ times 
with $C_{I}$ given in Eq.(\ref{5301}).
\end{theorem}

\begin{proof}
Since the Hamiltonian of the system is 
invariant under permutations. The wave function $|\Phi\rangle$ forms a base 
of the space $D$ that carries the representation of $S_{N}$.
By using Theorem (\ref{theorem1}), one can find that if 
the canonical partition function is written in the form Eq. (\ref{5301}),
then the space $D$ can be decomposed into subspace $V^{I}$ that carries the equivalent
and irreducible representation of $S_{N}$, and $V^{I}$ occurs $C_{I}$ times with $C_{I}$
given in Eq.(\ref{5301}). Thus, Eq. (\ref{pro6}) holds.
\end{proof}

It is obvious that the permutation phase given in Eq. (\ref{pro6}) is a matrix. 
For Bose and Fermi cases, the permutation phase,
$D\left(  \sigma_{ij}\right) $, recovers $\pm1$.

\subsection{The statistics with the Hamiltonian noninvariant under permutations}

A physics system is described by a Hamiltonian. The canonical partition 
function which is obtained by taking statistical averages, however, 
does not contain all the information of the physics system. 
Consequently, the invariance of the Hamiltonian under permutations 
ensures the invariance of the canonical partition function under 
permutations, but not vice versa.

In this section, we provide a method to distinguish statistics with
the Hamiltonian noninvariant under permutations. 
For such statistics, the permutation phase can not be constructed in the scheme.

\begin{theorem}
\label{theorem111} If the canonical partition function of an ideal gas consisting of $N$%
-identical particles, $Z\left(  \beta,N\right)  $, can not be written in
the form Eq. (\ref{5301}) with non-negative-integer coefficient, then the Hamiltonian 
is noninvariant under permutations, 
i.e., $\left[  H,S_{N}\right]  \neq0$.
\end{theorem}

\begin{proof}
Theorem (\ref{theorem111}) is the reverse proposition of Theorem (\ref{theorem1}). 
Thus Theorem (\ref{theorem111}) holds.
\end{proof}

Examples such as Gentile statistics will be given in the following section.

\section{A unified framework}

For Bose statistics, there is no limitation on the maximum occupation and
the wave function is symmetric with the permutation phase $e^{i0}=1$.
For Fermi statistics, the maximum occupation number is $1$ and 
the wave function is anti-symmetric with the permutation phase $e^{i\pi}=-1$.
The Hilbert subspace 
describing Bose statistics is the symmetric 
subspace, as shown in Fig.
(\ref{Bose}). The Hilbert subspace describing Fermi statistics 
is the anti-symmetric subspace, as shown in Fig. (\ref{Fermi}). 

In this section, as examples, 
we describe a series of generalized quantum statistics, such
as parastatistics \cite{green1953generalized,ohnuki1982quantum}, the
intermediate statistics or Gentile statistics
\cite{gentile1940itosservazioni,dai2004gentile}, Gentileonic statistics
\cite{cattani1984general}, and the immannons \cite{tichy2017extending}, in a
unified framework. For these kinds of statistics,
the canonical partition function, the maximum number, and the
permutation phase of the wave function 
are given. Especially, three new generalized quantum statistics, 
which seem to be the missing pieces in the puzzle are proposed. 

The Hilbert subspace helps to illustrate the difference 
between generalized quantum statistics intuitively. Therefore, we also
give the Hilbert subspace for these kinds of statistics.

\subsection{The $N$-distinguishable-particle gas system: Boltzmann statistics}

The system consisting of distinguishable particles obeys Boltzmann statistics. 
The canonical partition function for an ideal $N$-distinguishable-particle gas is
\cite{reichl2009modern,pathria2011statistical}
\begin{equation}
Z_{cl}\left(  \beta,N\right)  =\left(  \sum_{i}e^{-\beta\varepsilon_{i}%
}\right)  ^{N}.\label{uni1}%
\end{equation}
In this section, by discussing the $N$-distinguishable-particle gas in the scheme, 
we suggest an unconventional perspective to study the Hilbert subspace and the maximum
occupation number of the system. The distinguishablility of
the particle of Boltzmann statistics appears automatically.

\subsubsection{The Hilbert subspace}
One can verify that Eq. (\ref{uni1}) can be expressed as a linear combination of
the S-function:%
\begin{equation}
Z_{cl}\left(  \beta,N\right)  =\left(  \sum_{i}e^{-\beta\varepsilon_{i}%
}\right)  ^{N}=\sum_{I=1}^{P\left(  N\right)  }f_{I}s_{\left(  \lambda\right)
_{I}}\left(  e^{-\beta\varepsilon_{1}},e^{-\beta\varepsilon_{2}},...\right)
,\label{uni2}%
\end{equation}
where the coefficient $f_{I}$ is defined in Eq. (\ref{pro20}). By using
Theorem (\ref{theorem1}), Eq.
(\ref{uni2}) implies that the Hilbert space describing the $N$%
-distinguishable-particle system is
\begin{equation}
D=%
{\displaystyle\bigoplus\limits_{I=1}^{P\left(  N\right)  }}
\left(  V^{\prime I}\right)  ^{\oplus f_{I}}=%
{\displaystyle\bigoplus\limits_{I=1}^{P\left(  N\right)  }}
V^{I}=V^{\otimes N},
\end{equation}
where we use the fact that $V^{\prime I}$ occurs $f_{I}$ times in $V^{I}$,
i.e., $V^{I}=\left(  V^{\prime I}\right)  ^{\oplus f_{I}}$. Therefore, as shown
in Fig. (\ref{classical}), the  Hilbert space $V^{\otimes N}$ describes an
$N$-distinguishable-particle gas system.

\begin{figure}[H]
\centering
\includegraphics[width=0.8\textwidth]{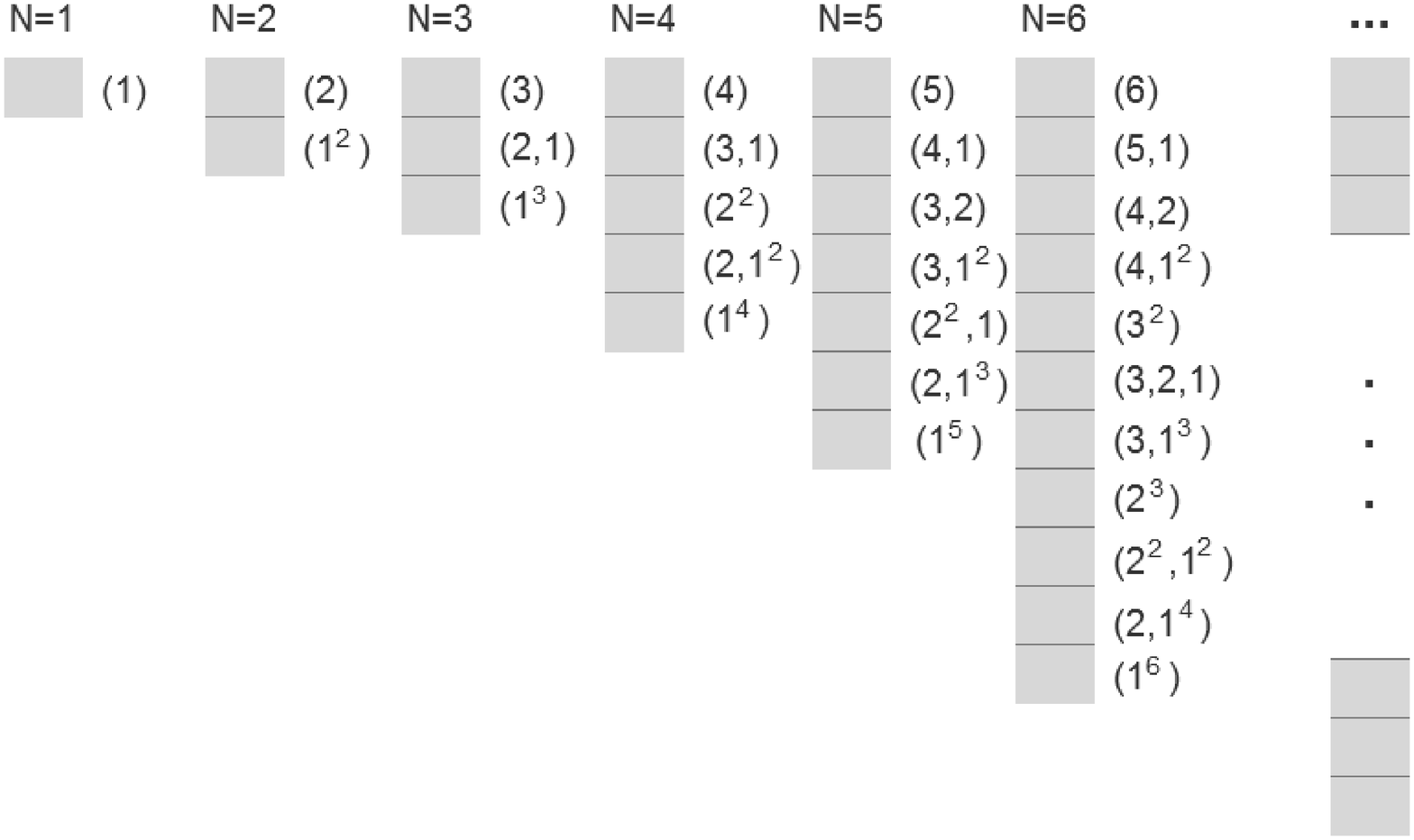}
\caption{The space $V^{\otimes N}$ describing the $N$-distinguishable-particle gas system. Each cube represents a subspace corresponding to an integer partition.}
\label{classical}
\end{figure}

\begin{figure}[H]
\centering
\includegraphics[width=0.8\textwidth]{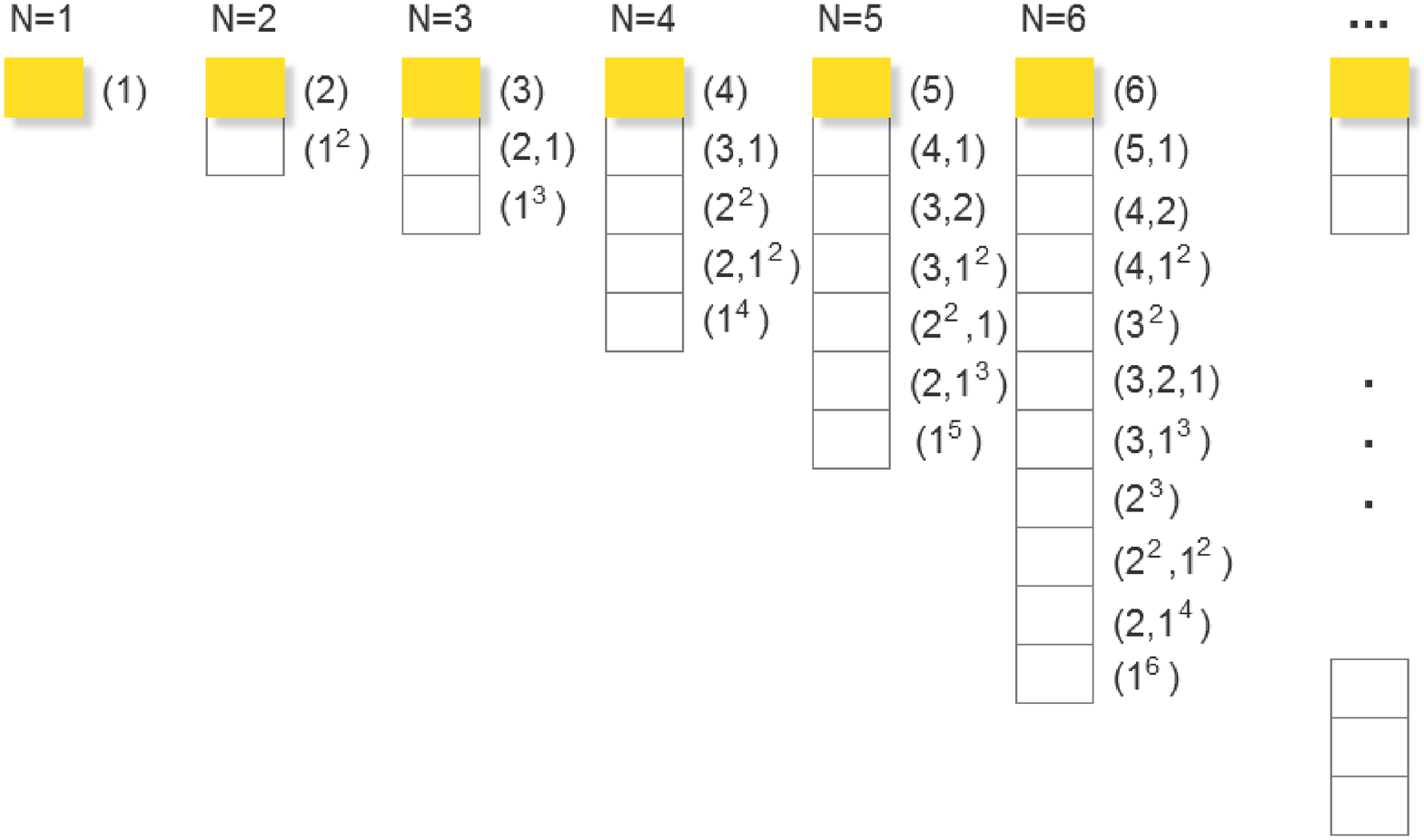}
\caption{The Hilbert subspace describing Bose-Einstein statistics.}
\label{Bose}
\end{figure}

\begin{figure}[H]
\centering
\includegraphics[width=0.8\textwidth]{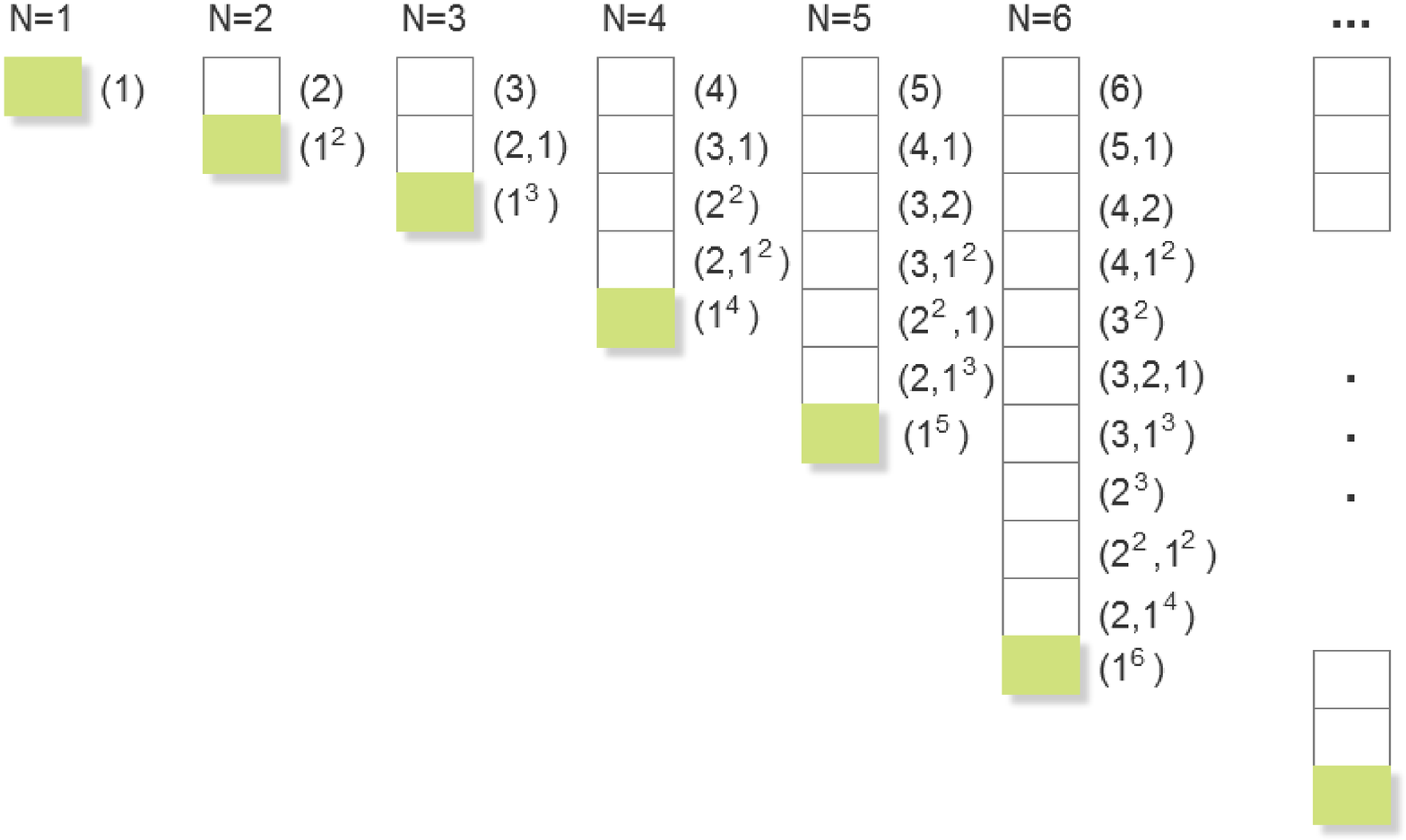}
\caption{The Hilbert subspace describing Fermi-Dirac statistics.}
\label{Fermi}
\end{figure}

\subsubsection{The maximum occupation number and the distinguishablility of
the particle.}

One can verify that Eq. (\ref{uni1}) can be expressed as a linear combination of
the monomial symmetric polynomial $m_{\left(  \lambda\right)  }\left(
e^{-\beta\varepsilon_{1}},e^{-\beta\varepsilon_{2}},\ldots\right)  $:%
\begin{equation}
Z_{cl}\left(  \beta,N\right)  =\left(  \sum_{i}e^{-\beta\varepsilon_{i}%
}\right)  ^{N}=\sum_{I=1}^{P\left(  N\right)  }z_{I}m_{\left(  \lambda\right)
_{I}}\left(  e^{-\beta\varepsilon_{1}},e^{-\beta\varepsilon_{2}},...\right)  ,
\label{dis0}%
\end{equation}
where $z_{I}$ satisfies%

\begin{equation}
z_{I}=N!\left(
{\displaystyle\prod\limits_{j=1}^{N}}
\lambda_{I,j}!\right)  ^{-1}.\label{dis1}%
\end{equation}

By using Theorem (\ref{theorem2}), one can find that there is no limitation on the maximum
occupation number for an $N$-distinguishable-particle gas system, since the first term in Eqs. (\ref{dis0}) is $m_{\left(  N\right)  }\left(
e^{-\beta\varepsilon_{1}},e^{-\beta\varepsilon_{2}},...\right)  $.

A direct manifestation of the distinguishablility of
the particle is given: the coefficient $z_{I}$ appears because the number of microstates with
$\lambda_{I,1}$ distinguishable particles occupying a quantum state,
$\lambda_{I,2}$ distinguishable particles occupying another quantum state,
and so on, is exactly $z_{I}$.

\subsection{Gentile statistics}

Gentile statistics is a generalization of Bose-Einstein and Fermi-Dirac statistics. The
maximum occupation number of Gentile statistics is an integer $q$
\cite{gentile1940itosservazioni,dai2004gentile,maslov2017relationship}.

The canonical partition function of an ideal Gentile gas is
\cite{zhou2018canonical}
\begin{equation}
Z_{q}^{G}\left(  \beta,N\right)  =\sum_{I=1}^{P\left(  N\right)  \ }Q^{I}\left(
q\right)  s_{\left(  \lambda\right)  _{I}}\left(  e^{-\beta\varepsilon_{1}%
},e^{-\beta\varepsilon_{2}},...\right)  , \label{gen1}%
\end{equation}
where the coefficient $Q^{I}\left(  q\right)  $ is given as
\begin{equation}
\sum_{J=1}^{P\left(  N\right)  }\Gamma^{J}\left(  q\right)  \left(  k_{J}%
^{I}\right)  ^{-1}=Q^{J}\left(  q\right)  . \label{gen2}%
\end{equation}
$\left(  k_{J}^{I}\right)  ^{-1}$ satisfies $\left(  k_{J}^{I}\right)
^{-1}k_{I}^{L}=\delta_{J}^{L}$ and $\Gamma^{J}\left(  q\right)  $ satisfies
\begin{equation}
\Gamma^{J}\left(  q\right)  =\left\{
\begin{array}
[c]{c}%
1,\text{if }\lambda_{I,1}\leqslant q\\
0\text{, otherwise}%
\end{array}
\right.  . \label{gen3}%
\end{equation}

In this section, we give a discussion on Gentile statistics in the scheme.
Especially, we show that the Hamiltonian of a Gentile-statistics system is noninvariant
under permutations.

\subsubsection{Discussions on the permutation phase of the wave function}

By using Theorem (\ref{theorem111}), one can find that, for Gentile statistics, 
the Hamiltonian is noninvariant under permutations. 
The permutation phase can not be constructed in the scheme.
This can be verified, from Eq. (\ref{gen1}), that the coefficient, $Q^{I}\left(  q\right)  $, is not
nonnegative. As examples, we list the explicit expression of the canonical partition
function of $N=3$, $4$, $5$, and $6$ with various maximum occupation
numbers $q$.

$N=3$,
\begin{equation}
Z_{2}^{G}\left(  \beta,3\right)  =s_{\left(  2,1\right)  }-s_{\left(
1^{3}\right)  }, \label{gen4}%
\end{equation}
where we denote $s_{\left(  \lambda\right)  }=s_{\left(  \lambda\right)
}\left(  e^{-\beta\varepsilon_{1}},e^{-\beta\varepsilon_{2}},...\right)  $ for convenience.

$N=4$,%
\begin{equation}
Z_{2}^{G}\left(  \beta,4\right)  =s_{\left(  2^{2}\right)  }-s_{\left(
1^{4}\right)  },
\end{equation}%
\begin{equation}
Z_{3}^{G}\left(  \beta,4\right)  =s_{\left(  3,1\right)  }-s_{\left(
2,1^{2}\right)  }+s_{\left(  1^{4}\right)  }.
\end{equation}

$N=5$,%
\begin{equation}
Z_{2}^{G}\left(  \beta,5\right)  =s_{\left(  2^{2},1\right)  }-s_{\left(
2,1^{3}\right)  },
\end{equation}%
\begin{equation}
Z_{3}^{G}\left(  \beta,5\right)  =s_{\left(  3,2\right)  }-s_{\left(
2^{2},1\right)  }+s_{\left(  1^{5}\right)  },
\end{equation}%
\begin{equation}
Z_{4}^{G}\left(  \beta,5\right)  =s_{\left(  4,1\right)  }-s_{\left(
3,1^{2}\right)  }+s_{\left(  2,1^{3}\right)  }-s_{\left(  1^{5}\right)  }.
\end{equation}

$N=6$,%
\begin{equation}
Z_{2}^{G}\left(  \beta,6\right)  =s_{\left(  2^{3}\right)  }-s_{\left(
2,1^{4}\right)  }+s_{\left(  1^{6}\right)  },
\end{equation}%
\begin{equation}
Z_{3}^{G}\left(  \beta,6\right)  =s_{\left(  3,3\right)  }-s_{\left(  2^{2}%
,1^{2}\right)  }+s_{\left(  2,1^{4}\right)  },
\end{equation}%
\begin{equation}
Z_{4}^{G}\left(  \beta,6\right)  =s_{\left(  4,2\right)  }-s_{\left(
3,2,1\right)  }+s_{\left(  2^{2},1^{2}\right)  }-s_{\left(  1^{6}\right)  },
\end{equation}%
\begin{equation}
Z_{5}^{G}\left(  \beta,6\right)  =s_{\left(  5,1\right)  }-s_{\left(
4,1^{2}\right)  }+s_{\left(  3,1^{3}\right)  }-s_{\left(  2,1^{4}\right)
}+s_{\left(  1^{6}\right)  }. \label{gen6}%
\end{equation}

The Hilbert subspace describing Gentile statistics is the subspaces with
nonzero coefficient $Q^{I}\left(  q\right)  $, as shown in Figs.
(\ref{gentile2})-(\ref{gentile4}). 
One can find that the subspace becomes complicated as $N$ increases.

\begin{figure}[H]
\centering
\includegraphics[width=0.8\textwidth]{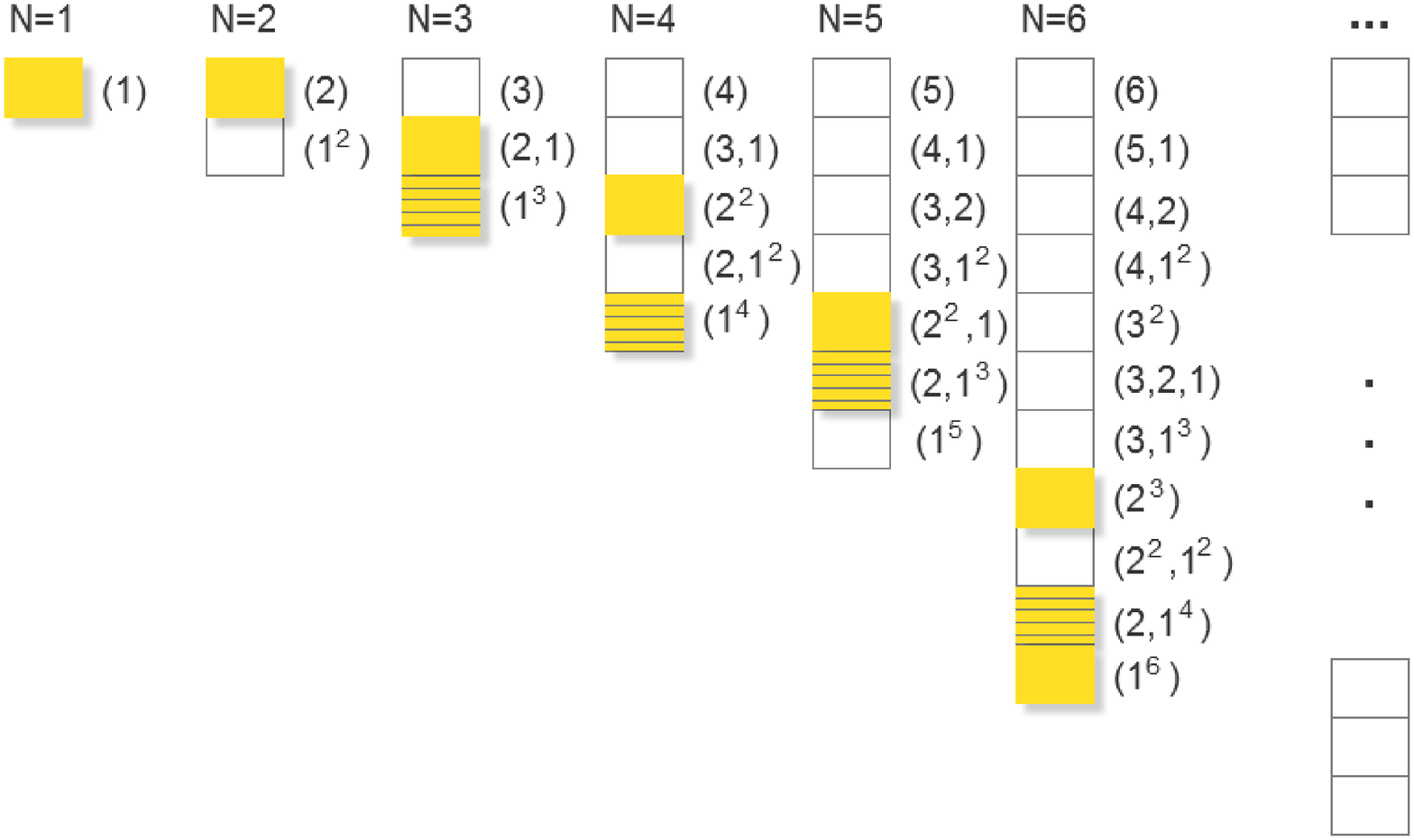}
\caption{The Hilbert subspace describing Gentile statistics with maximum occupation number q=2. The cube with horizontal lines means that the corresponding coefficient is negative.}
\label{gentile2}
\end{figure}

\begin{figure}[H]
\centering
\includegraphics[width=0.8\textwidth]{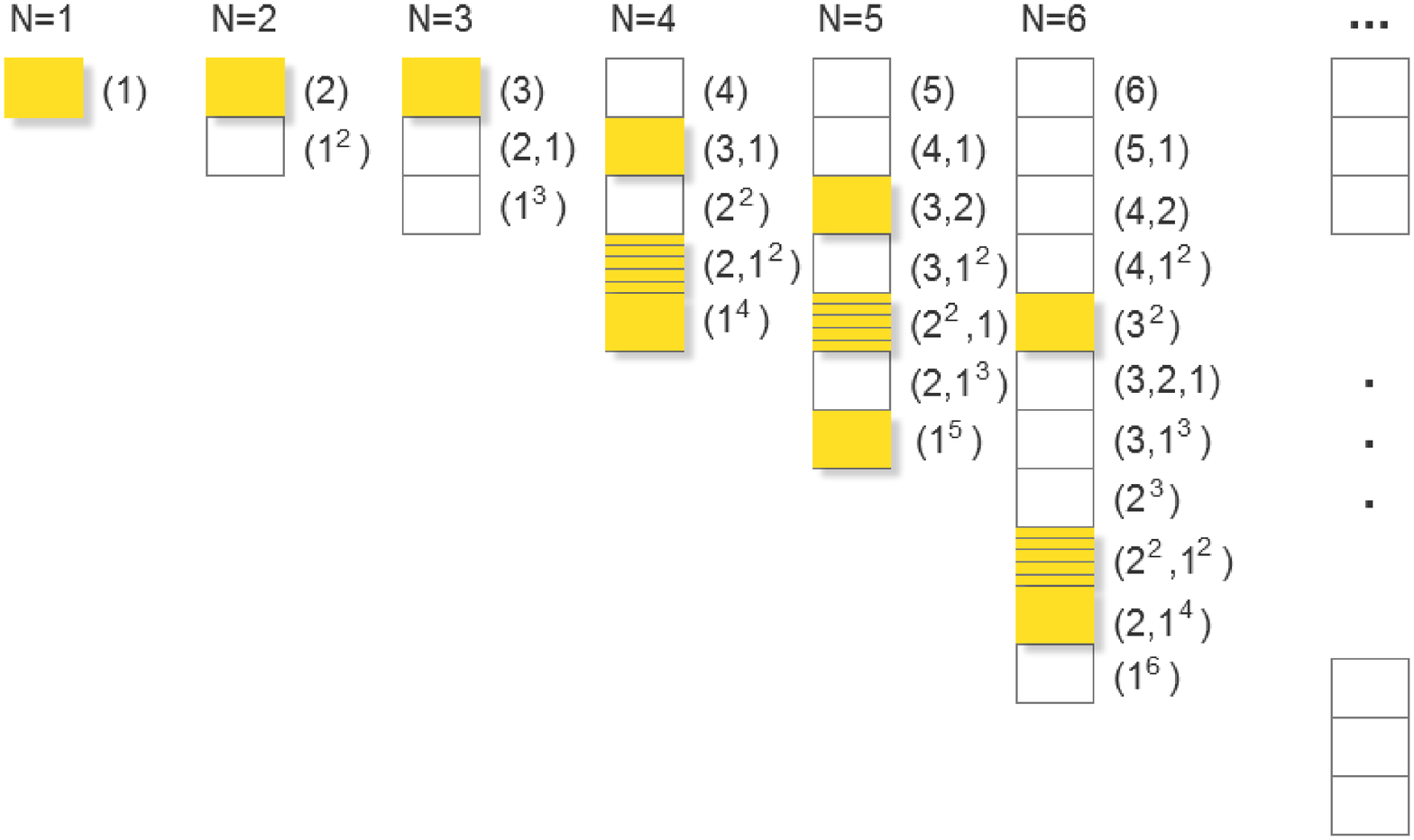}
\caption{The Hilbert subspace describing Gentile statistics with maximum occupation number q=3.}
\label{gentile3}
\end{figure}

\begin{figure}[H]
\centering
\includegraphics[width=0.8\textwidth]{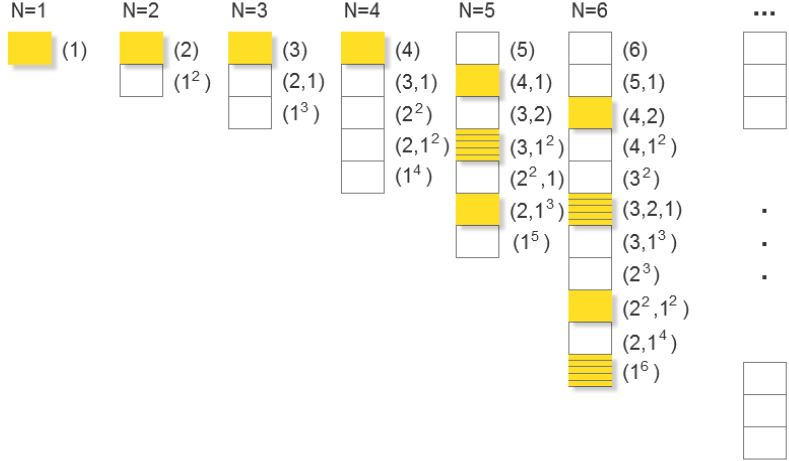}
\caption{The Hilbert subspace describing Gentile statistics with maximum occupation number q=4.}
\label{gentile4}
\end{figure}

\subsection{Parastatistics}

Parastatitics is proposed by H. Green as a generalization of Bose and Fermi
statistics in 1953 \cite{green1953generalized,cattani2009intermediate}. In
Green's generalization, a trilinear relation of the algebra of creation and
annihilation operators is proposed
\cite{chaturvedi1996canonical,chaturvedi1996grand}. Moreover, one already
knows that parastatitics corresponding to the higher dimensional
representation of the permutation group \cite{khare2005fractional}. For
instance, Okayama (1952) \cite{okayama1952generalization} suggests
that all irreducible representations associated with Young diagrams of at most
$p$ columns yield parastatitics \cite{vo1972first}.

The canonical partition function for para-Bose and para-Fermi statistics with $q$ the parameter is
\cite{chaturvedi1996canonical}%
\begin{align}
Z_{q}^{PB}\left(  \beta,N\right)   &  =\sum_{I=1,l_{\left(  \lambda\right)
}\leq q}^{P\left(  N\right)  }s_{\left(  \lambda\right)  _{I}}\left(
e^{-\beta\varepsilon_{1}},e^{-\beta\varepsilon_{2}},...\right)  ,
\label{para1}\\
Z_{q}^{PF}\left(  \beta,N\right)   &  =\sum_{I=1,\lambda_{I}\leq q}^{P\left(
N\right)  }s_{\left(  \lambda\right)  _{I}}\left(  e^{-\beta\varepsilon_{1}%
},e^{-\beta\varepsilon_{2}},...\right)  . \label{para2}%
\end{align}

In this section, we discuss parastatistics in the scheme. 
The permutation phase for parastatistics is discussed in 
Refs. \cite{khare2005fractional,okayama1952generalization,vo1972first}. Our method is 
new and gives results such as the dimension of the permutation phase of the wave function
corresponding to parastatistics.

\subsubsection{The maximum occupation number}

The canonical partition function of parastatistics, Eqs. (\ref{para1}) and
(\ref{para2}), can be rewritten in terms of $m_{\left(  \lambda\right)  _{I}%
}\left(  e^{-\beta\varepsilon_{1}},e^{-\beta\varepsilon_{2}},...\right)  $:%
\begin{align}
Z_{q}^{PB}\left(  \beta,N\right)   &  =\sum_{I,J=1}^{P\left(  N\right)  }%
P_{b}^{J}\left(  q\right)  k_{J}^{I}m_{\left(  \lambda\right)  _{I}}\left(
e^{-\beta\varepsilon_{1}},e^{-\beta\varepsilon_{2}},...\right)  ,
\label{para3}\\
Z_{q}^{PF}\left(  \beta,N\right)   &  =\sum_{I,J=1}^{P\left(  N\right)  }%
P_{f}^{J}\left(  q\right)  k_{J}^{I}m_{\left(  \lambda\right)  _{I}}\left(
e^{-\beta\varepsilon_{1}},e^{-\beta\varepsilon_{2}},...\right)  ,
\label{para4}%
\end{align}
where%
\begin{align}
P_{b}^{J}\left(  q\right)   &  =\left\{
\begin{array}
[c]{c}%
1\text{, for }l_{\left(  \lambda\right)  _{J}}\leq q,\\
0\text{, otherwise}%
\end{array}
\right.  ,\\
P_{f}^{J}\left(  q\right)   &  =\left\{
\begin{array}
[c]{c}%
1\text{, for }\lambda_{J,1}\leq q,\\
0\text{, otherwise}%
\end{array}
\right.  ,
\end{align}

Since the Kostka number is a lower triangular matrix
\cite{littlewood1977theory,macdonald1998symmetric}, for para-Fermi
statistics, the first term in the canonical partition function is
$m_{\left(  \lambda\right)  ^{q}}\left(  e^{-\beta\varepsilon_{1}}%
,e^{-\beta\varepsilon_{2}},...\right)  $ with $\lambda_{I,1}=q$. For the
para-Bose statistics, the first term in the canonical partition function is $m_{\left(  N\right)  }\left(  e^{-\beta\varepsilon_{1}},e^{-\beta
\varepsilon_{2}},...\right)  $ with $\lambda_{I,1}=N$. Eqs. (\ref{para3}) and
(\ref{para4}) can be written in the form%
\begin{align}
Z_{q}^{PB}\left(  \beta,N\right)   &  =m_{\left(  N\right)  }\left(
e^{-\beta\varepsilon_{1}},e^{-\beta\varepsilon_{2}},...\right)  +\sum
_{\lambda_{I,1}<N}^{P\left(  N\right)  }M^{\prime I}m_{\left(  \lambda\right)
_{I}}\left(  e^{-\beta\varepsilon_{1}},e^{-\beta\varepsilon_{2}},...\right)
,\\
Z_{q}^{PF}\left(  \beta,N\right)   &  =m_{\left(  \lambda\right)  ^{q}}\left(
e^{-\beta\varepsilon_{1}},e^{-\beta\varepsilon_{2}},...\right)  +\sum
_{\lambda_{I,1}<q}^{P\left(  N\right)  }M^{\prime\prime I}m_{\left(
\lambda\right)  _{I}}\left(  e^{-\beta\varepsilon_{1}},e^{-\beta
\varepsilon_{2}},...\right)  .
\end{align}
By using Theorem (\ref{theorem2}), one can find that for para-Fermi statistics with 
parameter $q$, the maximum
occupation number is $q$. For para-Bose statistics, there is no limitation on the maximum
occupation number.

For the sake of clarity, we list some of the explicit expression of the
canonical partition function for parastatistics of $N=3$, $4$, $5$ with
different $q$ (for $q=1$ parastatistics recovers Bose and Fermi
statistics). The detail of the calculation can be found in appendixes.

$N=3$, for para-Bose cases,%

\begin{equation}
Z_{2}^{PB}\left(  \beta,3\right)  =m_{\left(  3\right)  }+2m_{\left(
2,1\right)  }+3m_{\left(  1^{3}\right)  },\label{para001}%
\end{equation}

\begin{equation}
Z_{3}^{PB}\left(  \beta,3\right)  =m_{\left(  3\right)  }+2m_{\left(
2,1\right)  }+3m_{\left(  1^{3}\right)  }.
\end{equation}
For para-Fermi cases,%
\begin{equation}
Z_{2}^{PF}\left(  \beta,3\right)  =m_{\left(  2,1\right)  }+3m_{\left(
1^{3}\right)  },
\end{equation}

\begin{equation}
Z_{3}^{PF}\left(  \beta,3\right)  =m_{\left(  3\right)  }+2m_{\left(
2,1\right)  }+4m_{\left(  1^{3}\right)  }.
\end{equation}
$N=4$, for para-Bose cases,%

\begin{equation}
Z_{2}^{PB}\left(  \beta,4\right)  =m_{\left(  4\right)  }+2m_{\left(
3,1\right)  }+3m_{\left(  2^{2}\right)  }+4m_{\left(  2,1^{2}\right)
}+6m_{\left(  1^{4}\right)  },
\end{equation}%
\begin{equation}
Z_{3}^{PB}\left(  \beta,4\right)  =m_{\left(  4\right)  }+2m_{\left(
3,1\right)  }+3m_{\left(  2^{2}\right)  }+5m_{\left(  2,1^{2}\right)
}+9m_{\left(  1^{4}\right)  },
\end{equation}%
\begin{equation}
Z_{4}^{PB}\left(  \beta,4\right)  =m_{\left(  4\right)  }+2m_{\left(
3,1\right)  }+3m_{\left(  2^{2}\right)  }+5m_{\left(  2,1^{2}\right)
}+10m_{\left(  1^{4}\right)  }.
\end{equation}
For para-Fermi cases,%
\begin{equation}
Z_{2}^{PF}\left(  \beta,4\right)  =m_{\left(  2^{2}\right)  }+2m_{\left(
2,1^{2}\right)  }+6m_{\left(  1^{4}\right)  },
\end{equation}%
\begin{equation}
Z_{3}^{PF}\left(  \beta,4\right)  =m_{\left(  3,1\right)  }+2m_{\left(
2^{2}\right)  }+4m_{\left(  2,1^{2}\right)  }+9m_{\left(  1^{4}\right)  },
\end{equation}

\begin{equation}
Z_{4}^{PF}\left(  \beta,4\right)  =m_{\left(  4\right)  }+2m_{\left(
3,1\right)  }+3m_{\left(  2^{2}\right)  }+5m_{\left(  2,1^{2}\right)
}+10m_{\left(  1^{4}\right)  }.
\end{equation}
$N=5$, for para-Bose cases,%

\begin{align}
Z_{2}^{PB}\left(  \beta,5\right)   &  =m_{\left(  5\right)  }+2m_{\left(
4,1\right)  }+3m_{\left(  3,2\right)  }+4m_{\left(  3,1^{2}\right)
}\nonumber\\
&  +5m_{\left(  2^{2},1\right)  }+7m_{\left(  2,1^{3}\right)  }+10m_{\left(
1^{5}\right)  },
\end{align}%
\begin{align}
Z_{3}^{PB}\left(  \beta,5\right)   &  =m_{\left(  5\right)  }+2m_{\left(
4,1\right)  }+3m_{\left(  3,2\right)  }+5m_{\left(  3,1^{2}\right)
}\nonumber\\
&  +7m_{\left(  2^{2},1\right)  }+12m_{\left(  2,1^{3}\right)  }+21m_{\left(
1^{5}\right)  },
\end{align}%
\begin{align}
Z_{4}^{PB}\left(  \beta,5\right)   &  =m_{\left(  5\right)  }+2m_{\left(
4,1\right)  }+3m_{\left(  3,2\right)  }+5m_{\left(  3,1^{2}\right)
}\nonumber\\
&  +7m_{\left(  2^{2},1\right)  }+13m_{\left(  2,1^{3}\right)  }+25m_{\left(
1^{5}\right)  },
\end{align}%
\begin{align}
Z_{5}^{PB}\left(  \beta,5\right)   &  =m_{\left(  5\right)  }+2m_{\left(
4,1\right)  }+3m_{\left(  3,2\right)  }+5m_{\left(  3,1^{2}\right)
}\nonumber\\
&  +7m_{\left(  2^{2},1\right)  }+13m_{\left(  2,1^{3}\right)  }+26m_{\left(
1^{5}\right)  }.
\end{align}
For para-Fermi cases,%
\begin{equation}
Z_{2}^{PF}\left(  \beta,5\right)  =m_{\left(  2^{2},1\right)  }+3m_{\left(
2,1^{3}\right)  }+10m_{\left(  1^{5}\right)  },
\end{equation}

\begin{equation}
Z_{3}^{PF}\left(  \beta,5\right)  =m_{\left(  3,2\right)  }+2m_{\left(
3,1^{2}\right)  }+4m_{\left(  2^{2},1\right)  }+9m_{\left(  2,1^{3}\right)
}+21m_{\left(  1^{5}\right)  },
\end{equation}

\begin{align}
Z_{4}^{PF}\left(  \beta,5\right)   &  =m_{\left(  4,1\right)  }+2m_{\left(
3,2\right)  }+4m_{\left(  3,1^{2}\right)  }+6m_{\left(  2^{2},1\right)
}\nonumber\\
&  +12m_{\left(  2,1^{3}\right)  }+25m_{\left(  1^{5}\right)  },
\end{align}

\begin{align}
Z_{5}^{PF}\left(  \beta,5\right)   &  =m_{\left(  5\right)  }+2m_{\left(
4,1\right)  }+3m_{\left(  3,2\right)  }+5m_{\left(  3,1^{2}\right)
}\nonumber\\
&  +7m_{\left(  2^{2},1\right)  }+13m_{\left(  2,1^{3}\right)  }+26m_{\left(
1^{5}\right)  }. \label{para02}%
\end{align}

From Eqs. (\ref{para001}) - (\ref{para02}), one can see that the coefficient
distinguishes parastatistics from Gentile statistics. For
example, for para-Fermi statistics with parameter $q$, although the
maximum occupation number is $q$, however, it does not yield Gentile
statistics, because the weight, represented by the coefficient, 
distinguishes those two kinds of statistics, e.g., $Z_{2}^{PF}\left(  \beta,5\right)  =m_{\left(  2^{2},1\right)
}+3m_{\left(  2,1^{3}\right)  }+10m_{\left(  1^{5}\right)  }$ and
$Z_{2}^{G}\left(  \beta,5\right)  =m_{\left(  2^{2},1\right)  }+m_{\left(
2,1^{3}\right)  }+m_{\left(  1^{5}\right)  }$. In Gentile statistics, the
microstate with the same occupation number only be counted once. 

\subsubsection{The permutation phase of the wave function}

The canonical partition function of 
parastatitics, Eqs. (\ref{para1}) and
(\ref{para2}), shows, using Theorem (\ref{theorem1}), that the Hilbert subspace describing para-Bose
statistics is a direct sum of those spaces corresponding to $\left(
\lambda\right)  $ with length smaller than $q$, as shown in Figs. (\ref{parab2}) and
(\ref{parab3}). The Hilbert subspace
describing para-Fermi statistics is a direct sum of those spaces corresponding to $\left(
\lambda\right)  $ with $\lambda_{1}$ smaller than $q$, as shown in Figs. (\ref{paraf2}) and
(\ref{paraf3}). That is,
\begin{align*}
&
{\displaystyle\bigoplus\limits_{l_{\left(  \lambda\right)  }\leq q}}
V^{\left(  \lambda\right)  }\text{ for para-Bose statistics with parameter
}q\text{,}\\
&
{\displaystyle\bigoplus\limits_{\lambda_{I,1}\leq q}}
V^{\left(  \lambda\right)  }\text{ for para-Fermi statistics with
parameter }q\text{.}%
\end{align*}
\begin{figure}[H]
\centering
\includegraphics[width=0.8\textwidth]{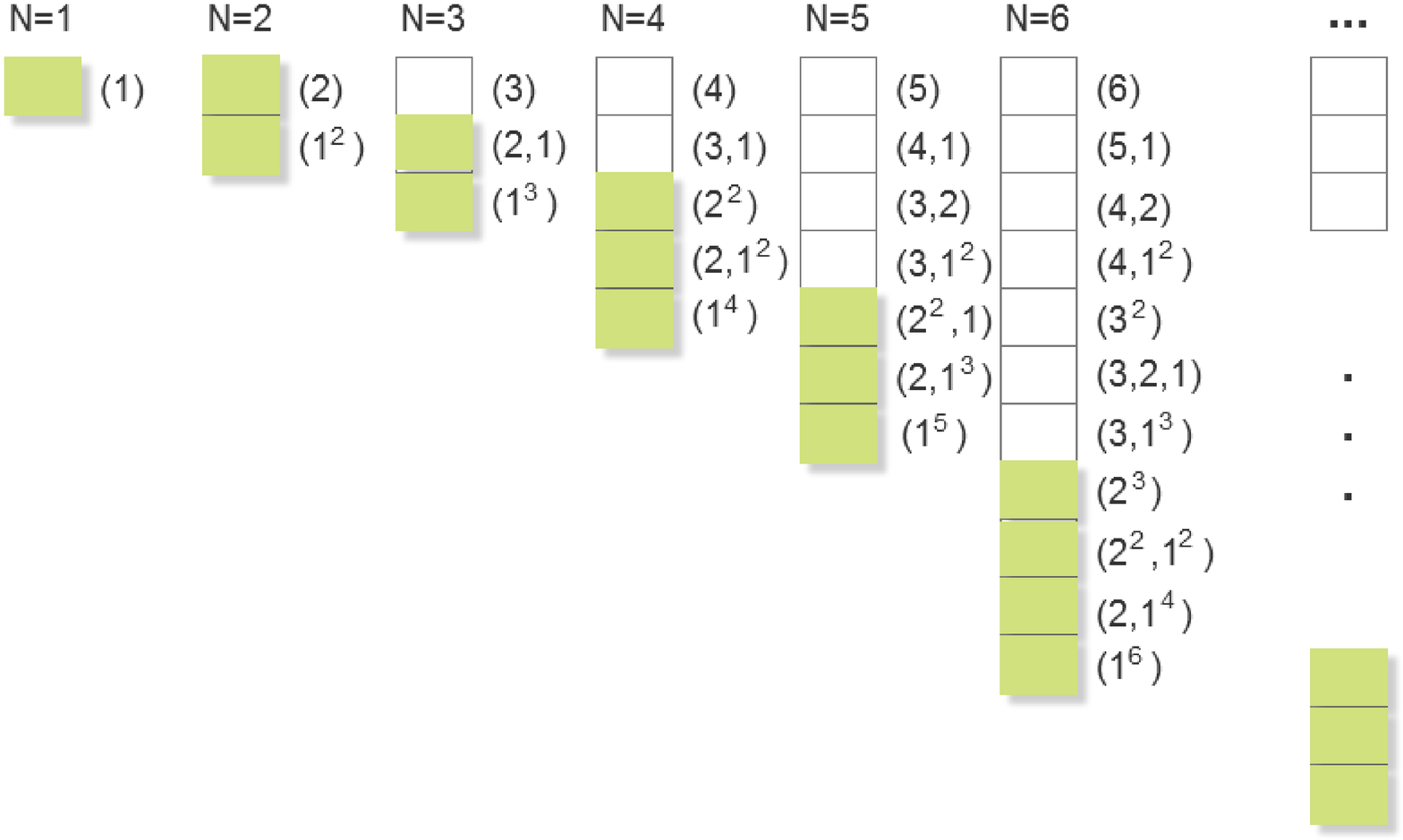}
\caption{The Hilbert subspace describing para-Fermi statistics with q=2.}
\label{paraf2}
\end{figure}

\begin{figure}[H]
\centering
\includegraphics[width=0.8\textwidth]{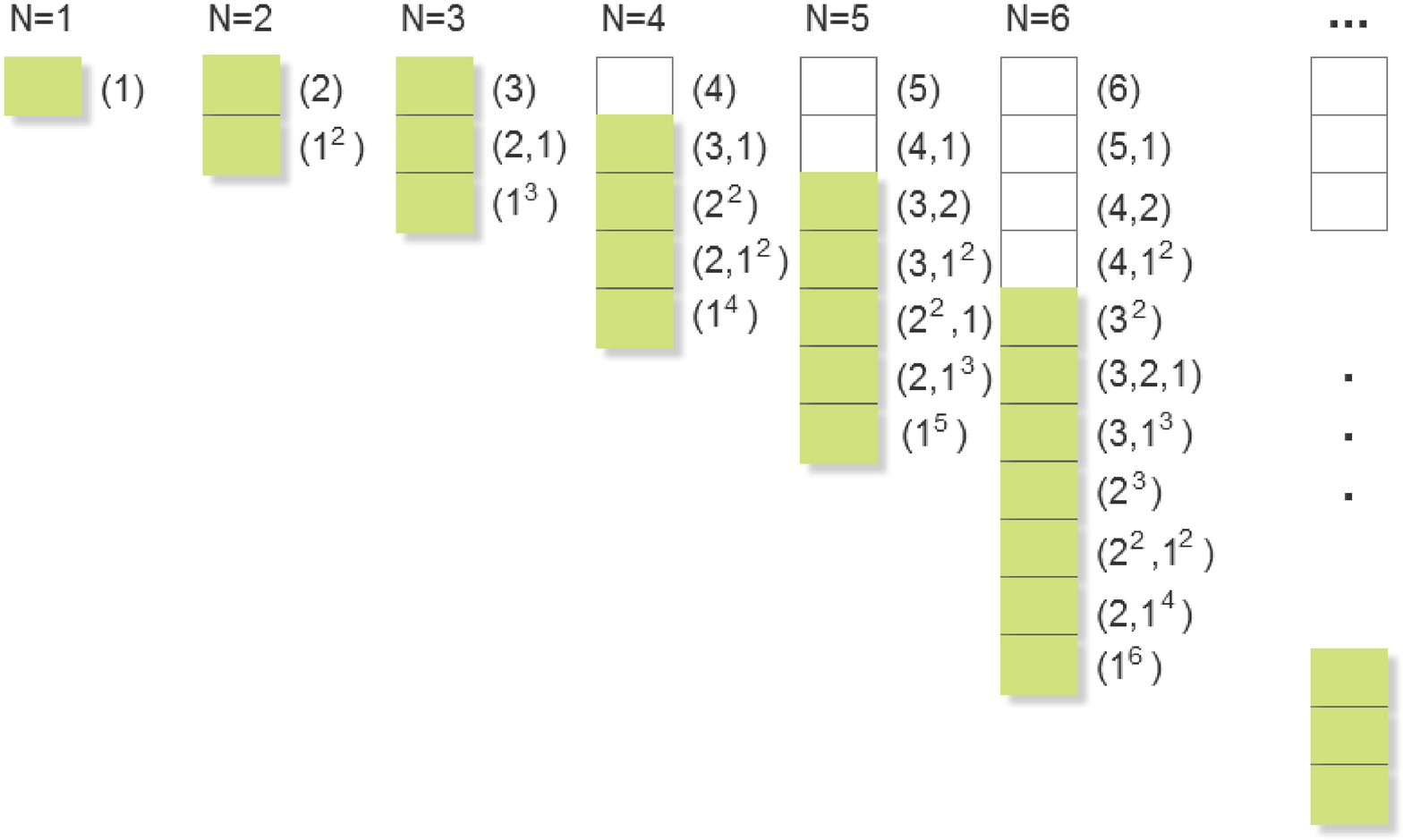}
\caption{The Hilbert subspace describing para-Fermi statistics with q=3.}
\label{paraf3}
\end{figure}
\begin{figure}[H]
\centering
\includegraphics[width=0.8\textwidth]{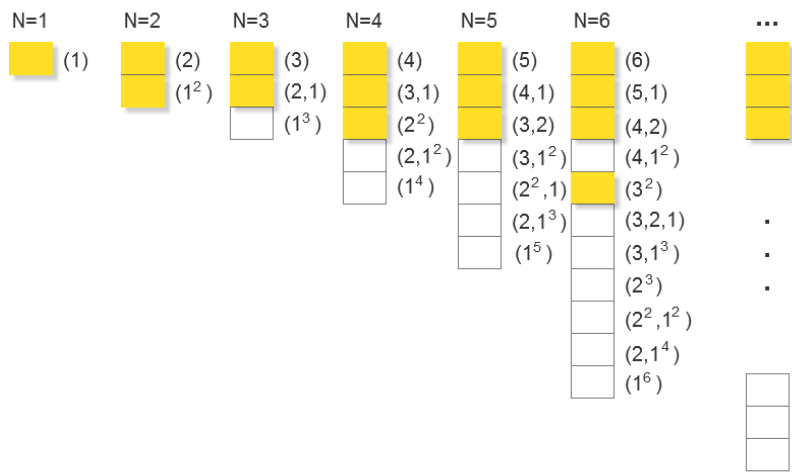}
\caption{The Hilbert subspace describing para-Bose statistics with q=2.}
\label{parab2}
\end{figure}

\begin{figure}[H]
\centering
\includegraphics[width=0.8\textwidth]{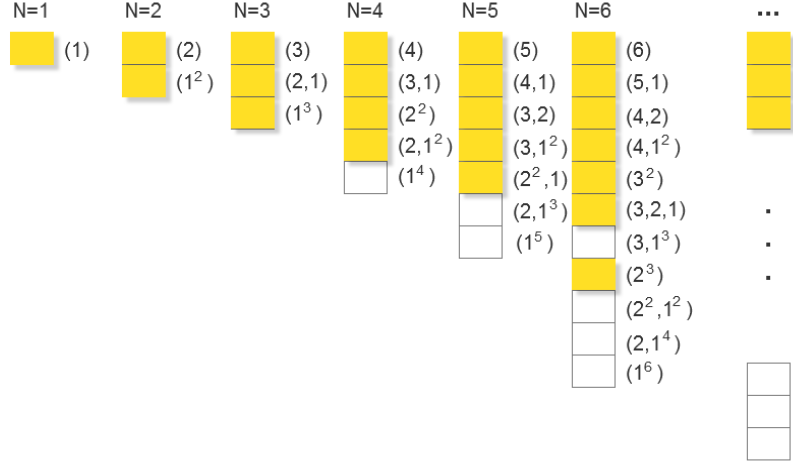}
\caption{The subspace describing para-Bose statistics with q=3.}
\label{parab3}
\end{figure}

For parastatistics, using Theorem (\ref{theorem11}),
the permutation phase of the wave function can be given as%
\begin{equation}
\sigma_{ij}|\Phi\rangle=D\left(  \sigma_{ij}\right)  |\Phi\rangle,
\end{equation}
where
\begin{equation}
D=\left\{
\begin{array}
[c]{c}%
{\displaystyle\bigoplus\limits_{l_{\left(  \lambda\right)  }\leq q}}
D^{I}\text{ for para-Bose statistics with parameter }q,\\%
{\displaystyle\bigoplus\limits_{\lambda_{I,1}\leq q}}
D^{I}\text{ for para-Fermi statistics with parameter }q.
\end{array}
\right.
\end{equation}
This is a multi-dimensional representation. By using the dimension
of the representation of the permutation group \cite{iachello2006lie},
we give the dimension of the permutation phase.
For para-Bose statistics with
parameter $q$,%

\begin{equation}
\dim\left(  D\right)  =\sum_{l_{\left(  \lambda\right)  _{I}}\leq q}N!%
{\displaystyle\prod\limits_{i=1,i<j}}
\left(  \lambda_{I,i}-\lambda_{I,j}-i+j\right)
{\displaystyle\prod\limits_{i=1}}
\left[  \left(  l_{\left(  \lambda\right)  }+\lambda_{I,i}-i\right)  !\right]
^{-1}.
\end{equation}
For para-Fermi statistics with parameter $q$,
\begin{equation}
\dim\left(  D\right)  =\sum_{\lambda_{I,1}\leq q}N!%
{\displaystyle\prod\limits_{i=1,i<j}}
\left(  \lambda_{I,i}-\lambda_{I,j}-i+j\right)
{\displaystyle\prod\limits_{i=1}}
\left[  \left(  l_{\left(  \lambda\right)  }+\lambda_{I,i}-i\right)  !\right]
^{-1}.
\end{equation}

\subsection{The immannons and Gentileonic statistics}

Gentileonic statistics is proposed by Cattani and Fernandes in 1984
\cite{cattani1984general}. The immannons is proposed by Tichy in 2017
\cite{tichy2017extending}. They both are generalized quantum statistics
corresponding to the higher dimensional representation of permutation groups
\cite{cattani1984general,tichy2017extending}.

In this section, we discuss the immannons and Gentileonic statistics in the scheme.
The result shows that in statistical mechanics, the
immannons and Gentileonic statistics are essentially equivalent.

\subsubsection{The canonical partition function}

\begin{theorem}
\label{theorem3} The canonical partition function of the immanonns and
Gentileonic statistics is%
\begin{equation}
Z^{IG}_{\left(  \lambda\right)  }\left(  \beta,N\right)  =s_{\left(
\lambda\right)  }\left(  e^{-\beta\varepsilon_{1}},e^{-\beta\varepsilon_{2}%
},...\right).  \label{im01}%
\end{equation}

\end{theorem}

\begin{proof}
Gentileonic statistics is related to the higher dimensional
representation of $S_{N}$ and the wave function is give as
\cite{cattani1984general,cattani2009intermediate}
\begin{equation}
\left\vert \Psi_{\left(  \lambda\right)  }\right\rangle =\frac{1}{\sqrt{f_{I}%
}}\left(
\begin{array}
[c]{c}%
\Phi_{1}^{\left(  \lambda\right)  }\\
\Phi_{2}^{\left(  \lambda\right)  }\\
...\\
\Phi_{f_{I}}^{\left(  \lambda\right)  }%
\end{array}
\right)  , \label{im2}%
\end{equation}
where $\Phi_{i}^{\left(  \lambda\right)  }$ is the operator associated with
the Young shapes \cite{hamermesh1962group} corresponding to the integer
partition $\left(  \lambda\right)  $. Since the operator associate with the
Young shape is one of the constructions for the basis of the subspace that
carries the irreducible representation for $S_{N}$ \cite{hamermesh1962group},
that is, the subspace spanned by the wave function $\left\vert \Psi_{\left(
\lambda\right)  }\right\rangle $ is $V^{\left(  \lambda\right)  }$. The
canonical partition function is%
\begin{align}
Z^{IG}_{\left(  \lambda\right)  }\left(  \beta,N\right)   &  =\sum\left\langle
\Psi_{\left(  \lambda\right)  }\right\vert e^{-\beta H_{N}}\left\vert
\Psi_{\left(  \lambda\right)  }\right\rangle\\
& =\frac{1}{f_{I}}\sum_{i=1}%
^{f_{I}}\left\langle \Phi_{i}^{\left(  \lambda\right)  }\right\vert e^{-\beta
H_{N}}\left\vert \Phi_{i}^{\left(  \lambda\right)  }\right\rangle \nonumber\\
&  =s_{\left(  \lambda\right)  }\left(  e^{-\beta\varepsilon_{1}}%
,e^{-\beta\varepsilon_{2}},...\right)  , \label{im3}%
\end{align}
where $\left\langle \Phi_{i}^{\left(  \lambda\right)  }\right\vert e^{-\beta
H_{N}}\left\vert \Phi_{i}^{\left(  \lambda\right)  }\right\rangle
=\left\langle \Phi_{j}^{\left(  \lambda\right)  }\right\vert e^{-\beta H_{N}%
}\left\vert \Phi_{j}^{\left(  \lambda\right)  }\right\rangle =s_{\left(
\lambda\right)  }\left(  e^{-\beta\varepsilon_{1}},e^{-\beta\varepsilon_{2}%
},...\right)  $ is used. $\Phi_{i}^{\left(  \lambda\right)  }$ and
$\Phi_{j}^{\left(  \lambda\right)  }$ give the equivalent and irreducible
representation for $H_{N}$.

The immannons \cite{tichy2017extending} is a kind of generalized quantum statistics, of which,
the inner product of the wave function gives the immanant, labeled by an
integer partition $\left(  \lambda\right)  $ \cite{tichy2017extending}. It
recovers Bose-Einstein statistics for $\left(  \lambda\right)  =\left(  N\right)  $,
and Fermi-Dirac statistics for $\left(  \lambda\right)  =\left(  1^{N}\right)  $.
The wave function of the immannons is \cite{tichy2017extending}
\begin{equation}
\left\vert \Phi_{\left(  \lambda\right)  }\right\rangle \propto\hat
{P}_{\lambda}\left\vert \psi_{1},\psi_{2},\ldots,\psi_{N}\right\rangle ,
\label{im6}%
\end{equation}
where
\begin{equation}
\hat{P}_{\lambda}=\frac{\chi_{\left(  \lambda\right)  }\left(  e\right)  }%
{N!}\sum_{\sigma\in S_{N}}\chi_{\left(  \lambda\right)  }\left(
\sigma\right)  \hat{Q}\left(  \sigma\right)  \label{im0}%
\end{equation}
with $\chi_{\left(  \lambda\right)  }\left(  \sigma\right)  $ the simple
characteristic of $\sigma$ and $\hat{Q}\left(  \sigma\right)  $ an operator
satisfying \cite{tichy2017extending}
\begin{equation}
\hat{Q}\left(  \sigma\right)  \left\vert \psi_{1},\psi_{2},\ldots
,\psi_{N}\right\rangle =\left\vert \psi_{\sigma_{1}},\psi_{\sigma_{2}}%
,\ldots,\psi_{\sigma_{N}}\right\rangle.
\end{equation}
 The wave function, Eq. (\ref{im6}), is a
construction of the basis for the subspace $V^{\prime\left(  \lambda\right)
}$. Thus the wave function, Eq. (\ref{im6}), directly
yields
\begin{equation}
Z^{IG}_{\left(  \lambda\right)  }\left(  \beta,N\right)  =\left\langle
\Phi_{\left(  \lambda\right)  }\right\vert e^{-\beta H_{N}}\left\vert
\Phi_{\left(  \lambda\right)  }\right\rangle =s_{\left(  \lambda\right)
}\left(  e^{-\beta\varepsilon_{1}},e^{-\beta\varepsilon_{2}},...\right).
\end{equation}

\end{proof}

Eq. (\ref{im01}) shows that in statistical mechanics, the immanonns and Gentileonic statistics
share the same canonical partition function and thus are
essentially the same statistics.

\subsubsection{Discussions on the permutation phase of the wave function}

The Hilbert subspace describing Gentileonic statistics and the
immanonns are slightly difference: the subspace describing the
Gentileonic statistics is $V^{\left(  \lambda\right)  }$ and the subspace describing the immanonns
is $V^{\prime\left(  \lambda\right)  }$, i.e.,
\begin{equation}
D=\left\{
\begin{array}
[c]{c}%
V^{\left(  \lambda\right)  }\text{, for Gentileonic statistics,}\\
V^{\prime\left(  \lambda\right)  }\text{, for the immanonns.}%
\end{array}
\right.
\end{equation}

We have shown in the proof of Theorem (\ref{theorem1}) that the subspace
$V^{\prime\left(  \lambda\right)  }$ occurs $f_{\left(  \lambda\right)  }$
times in $V^{\left(  \lambda\right)  }$, thus these two subspaces are
essentially the same.

The immanonns does not possess a permutation symmetry, because
$V^{\prime\left(  \lambda\right)  }$, spanned by wave function Eq. (\ref{im6}),
does not carry a representation for $S_{N}$. However, for Gentileonic
statistics the permutation phase of the wave function is
\begin{equation}
\sigma_{ij}|\Phi\rangle=D^{I}\left(  \sigma_{ij}\right)  |\Phi\rangle.
\end{equation}

\begin{figure}[H]
\centering
\includegraphics[width=0.8\textwidth]{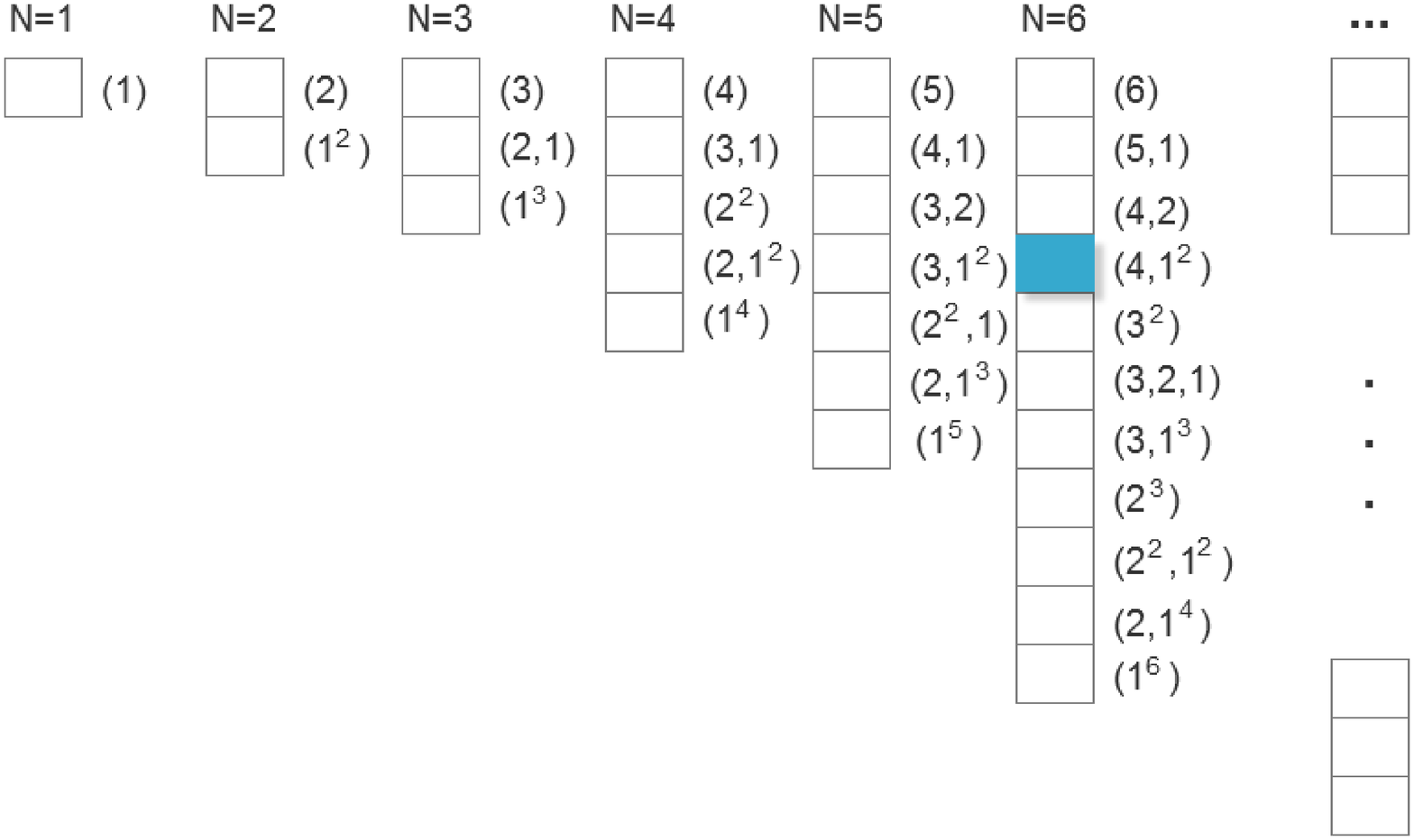}
\caption{The Hilbert subspace describing the immannons and Gentileonic statistics labeled by the integer partition (4,$1^{2}$).}
\label{imm}
\end{figure}

\subsubsection{The maximun occupation number}

We express the canonical partition function of the immannons and the
Gentileonic statistics, Eq. (\ref{im01}), in terms of $m_{\left(
\lambda\right)  _{I}}\left(  x_{1},x_{2},...\right)  $:
\begin{align}
Z_{\left(  \lambda\right)  _{J}}^{IG}\left(  \beta,N\right)   &  =\sum_{K=1}^{P\left(  N\right)  }\sum
_{I=1}^{P\left(  N\right)  }\delta_{J,K}k_{K}^{I}m_{\left(  \lambda\right)
_{I}}\left(  e^{-\beta\varepsilon_{1}},e^{-\beta\varepsilon_{2}},...\right)
\nonumber\\
&  =\sum_{I=1}^{P\left(  N\right)  }k_{J}^{I}m_{\left(  \lambda\right)  _{I}%
}\left(  e^{-\beta\varepsilon_{1}},e^{-\beta\varepsilon_{2}},...\right).
\label{im4}%
\end{align}
Since the Kostka number is a lower triangular matrix, i.e., $k_{J}^{I}=0$ if
$I<J$ \cite{littlewood1977theory,macdonald1998symmetric}, the first term
of in Eq. (\ref{im4}) is always
\begin{equation}
Z_{\left(  \lambda\right)  _{J}}^{IG}\left(  \beta,N\right)  =m_{\left(  \lambda\right)  _{J}}\left(
e^{-\beta\varepsilon_{1}},e^{-\beta\varepsilon_{2}},...\right)  +\ldots,
\end{equation}
which, using Theorem (\ref{theorem2}), implies that the maximum occupation number for the immanonns and Gentileonic
statistics is $\lambda_{J,1}$.

For example, the explicit expression of the canonical partition function for
the immannons and Gentileonic statistics labeled by the integer partition
$\left(  3,1^{2}\right)  $ is%
\begin{equation}
Z_{\left(  3,1^{2}\right)  }^{IG}\left(  \beta,N\right)  =m_{\left(  3,1^{2}%
\right)  }+m_{\left(  2^{2},1\right)  }+3m_{\left(  2,1^{3}\right)
}+6m_{\left(  1^{5}\right)  }.
\end{equation}
The coefficient in Eq. (\ref{im4}) distinguishes the immanonns and Gentileonic
statistics from Gentile statistics.

\subsection{generalized quantum statistics dual to Gentile statistics:
GD-ons}

In this section, we propose a new kind of generalized quantum statistics dual
to Gentile statistics. For convenience, we denote the particle GD-ons. 
By dual we mean that the
relation between GD-ons and Gentile statistics is an analog to the
relation between para-Bose and para-Fermi statistics.

\subsubsection{The canonical partition function}

The canonical partition function for GD-ons labeled by the integer
partition $\left(  \lambda\right)$ is%

\begin{equation}
Z_{q}^{\ast}\left(  N,\beta\right)  \equiv\sum_{I=1,l_{\left(  \lambda\right)
_{I}}\leq q}^{P\left(  N\right)  }m_{\left(  \lambda\right)  _{I}}\left(
e^{-\beta\varepsilon_{1}},e^{-\beta\varepsilon_{2}},...\right)  . \label{gd1}%
\end{equation}
GD-ons recovers bosons when $q=N$.

For Gentile statistics, the canonical partition function is given in Eq.
(\ref{maximum10}). From Eqs.
(\ref{gd1}) and (\ref{maximum10}), one can see that the relation between
GD-ons and Gentile statistics is an analog to the relation between
para-Bose and para-Fermi statistics of which the canonical partition function
is given in Eqs. (\ref{para1}) and (\ref{para2}).

\subsubsection{Discussions on the permutation phase of the wave
function}

Expressing the canonical partition function, Eq. (\ref{gd1}), in terms of the
S-function gives
\begin{equation}
Z_{q}^{\ast}\left(  N,\beta\right)  =\sum_{J=1}^{P\left(  N\right)  }%
G^{I}\left(  q\right)  s_{\left(  \lambda\right)  _{I}}\left(  e^{-\beta
\varepsilon_{1}},e^{-\beta\varepsilon_{2}},...\right)  ,\label{gd2}%
\end{equation}
where the coefficient $G^{I}\left(  q\right)  $ is given as
\begin{equation}
\sum_{J=1}^{P\left(  N\right)  }\Gamma^{\prime J}\left(  q\right)  \left(
k_{J}^{I}\right)  ^{-1}=G^{I}\left(  q\right)
\end{equation}
with
\begin{equation}
\Gamma^{\prime J}\left(  q\right)  =\left\{
\begin{array}
[c]{c}%
1,\text{if }l_{\left(  \lambda\right)  _{I}}\leq q\\
0\text{, otherwise}%
\end{array}
\right.  .
\end{equation}

By calculating the coefficient $G^{I}\left(  q\right)  $, we find that
$G^{I}\left(  q\right)  $ is not nonnegative. For example, the explicit
expressions for the canonical partition function of GD-ons with $N=5$ are%

\begin{equation}
Z_{1}^{\ast}\left(  5,\beta\right)  =s_{\left(  5\right)  }-s_{\left(
4,1\right)  }+s_{\left(  3,1^{2}\right)  }-s_{\left(  2,1^{3}\right)
}+s_{\left(  1^{5}\right)  },
\end{equation}%
\begin{equation}
Z_{2}^{\ast}\left(  5,\beta\right)  =s_{\left(  5\right)  }-s_{\left(
4,1\right)  }+2s_{\left(  2,1^{3}\right)  }-3s_{\left(  1^{5}\right)  },
\end{equation}%
\begin{equation}
Z_{3}^{\ast}\left(  5,\beta\right)  =s_{\left(  5\right)  }-s_{\left(
2,1^{3}\right)  }+3s_{\left(  1^{5}\right)  },
\end{equation}%
\begin{equation}
Z_{4}^{\ast}\left(  5,\beta\right)  =s_{\left(  5\right)  }-s_{\left(
1^{5}\right)  },
\end{equation}%
\begin{equation}
Z_{5}^{\ast}\left(  5,\beta\right)  =s_{\left(  5\right)  }+s_{\left(
4,1\right)  }+s_{\left(  3,2\right)  }+s_{\left(  3,1^{2}\right)  }+s_{\left(
2^{2},1\right)  }+s_{\left(  2,1^{3}\right)  }+s_{\left(  1^{5}\right)  }.
\end{equation}
Thus, for GD-ons, using Theorem (\ref{theorem111}), the Hamiltonian is noninvariant under 
permutations. The permutation phase of GD-ons can not be constructed
in the scheme.

The Hilbert subspace describing GD-ons is
the space with nonzero coefficients.

\subsubsection{The maximum occupation number}

Since the first term in the canonical partition function of GD-ons, Eqs.
(\ref{gd1}), is $m_{\left(  N\right)  }\left(  e^{-\beta\varepsilon_{1}%
},e^{-\beta\varepsilon_{2}},...\right)  $, i.e.,%
\begin{equation}
Z_{q}^{\ast}\left(  N,\beta\right)  =m_{\left(  N\right)  }\left(
e^{-\beta\varepsilon_{1}},e^{-\beta\varepsilon_{2}},...\right)  +\sum
_{I=2,l_{\left(  \lambda\right)  _{I}}\leq q}^{P\left(  N\right)  }m_{\left(
\lambda\right)  _{I}}\left(  e^{-\beta\varepsilon_{1}},e^{-\beta
\varepsilon_{2}},...\right) ,
\end{equation}
the maximum occupation number, using Theorem (\ref{theorem2}), is $N$, 
that is, there is no limitation on the
maximum occupation number for GD-ons.

\subsection{generalized quantum statistics corresponding to the monomial symmetric
function: M-ons}

We have shown that the S-function is closely related to the permutation phase
and the monomial symmetric function is closely related to the maximum
occupation number. The generalized quantum statistics corresponding to the S-function
is given as the immannons and Gentileonic statistics. In this section, 
we propose a kind of generalized quantum statistics corresponding to the
monomial symmetric function $m_{\left(  \lambda\right)}$. For the sake
of convenience, we denote the particle M-ons.

\subsubsection{The canonical partition function}

The canonical partition function of M-ons labeled by the integer
partition $\left(  \lambda\right)  $ is
\begin{equation}
Z_{\left(  \lambda\right)  }^{M}\left(  \beta,N\right)  =m_{\left(
\lambda\right)  }\left(  e^{-\beta\varepsilon_{1}},e^{-\beta\varepsilon_{2}%
},...\right)  . \label{M1}%
\end{equation}
The M-ons recover fermions when $\left(  \lambda\right)  =\left(
1^{N}\right)  $.

\subsubsection{Discussions on the permutation phase of the wave
function}

Expressing the canonical partition function, Eq.(\ref{M1}), in terms of the S-function gives
\begin{align}
Z_{\left(  \lambda\right)  _{J}}^{M}\left(  \beta,N\right)   &  =\sum
_{K=1}^{P\left(  N\right)  }\sum_{I=1}^{P\left(  N\right)  }\delta
_{J,K}\left(  k_{K}^{I}\right)  ^{-1}s_{\left(  \lambda\right)  _{I}}\left(
e^{-\beta\varepsilon_{1}},e^{-\beta\varepsilon_{2}},...\right) \nonumber\\
&  =\sum_{I=1}^{P\left(  N\right)  }\left(  k_{J}^{I}\right)  ^{-1}s_{\left(
\lambda\right)  _{I}}\left(  e^{-\beta\varepsilon_{1}},e^{-\beta
\varepsilon_{2}},...\right).  \label{M111}%
\end{align}

The coefficient in Eq. (\ref{M111}) is not nonnegative. For example, the
explicit expressions of the canonical partition function for M-ons labeled by
$\left(  \lambda\right)  =\left(  5\right)  $, $\left(  4,1\right)  $, ...,
are%
\begin{equation}
Z_{\left(  5\right)  }^{M}\left(  \beta,5\right)  =s_{\left(  5\right)
}-s_{\left(  4,1\right)  }+s_{\left(  3,1^{2}\right)  }-s_{\left(
2,1^{3}\right)  }+s_{\left(  1^{5}\right)  },\label{M2}%
\end{equation}%
\begin{equation}
Z_{\left(  4,1\right)  }^{M}\left(  \beta,5\right)  =s_{\left(  4,1\right)
}-s_{\left(  3,2\right)  }-s_{\left(  3,1^{2}\right)  }+s_{\left(
2^{2},1\right)  }+s_{\left(  2,1^{3}\right)  }-2s_{\left(  1^{5}\right)  },
\end{equation}%
\begin{equation}
Z_{\left(  3,2\right)  }^{M}\left(  \beta,5\right)  =s_{\left(  3,2\right)
}-s_{\left(  3,1^{2}\right)  }-s_{\left(  2^{2},1\right)  }+2s_{\left(
2,1^{3}\right)  }-2s_{\left(  1^{5}\right)  },
\end{equation}%
\begin{equation}
Z_{\left(  3,1^{2}\right)  }^{M}\left(  \beta,5\right)  =s_{\left(
3,1^{2}\right)  }-s_{\left(  2^{2},1\right)  }-s_{\left(  2,1^{3}\right)
}+3s_{\left(  1^{5}\right)  },
\end{equation}%
\begin{equation}
Z_{\left(  2^{2},1\right)  }^{M}\left(  \beta,5\right)  =s_{\left(
2^{2},1\right)  }-2s_{\left(  2,1^{3}\right)  }+3s_{\left(  1^{5}\right)  },
\end{equation}%
\begin{equation}
Z_{\left(  2,1^{3}\right)  }^{M}\left(  \beta,5\right)  =s_{\left(
2,1^{3}\right)  }-4s_{\left(  1^{5}\right)  }.\label{M3}%
\end{equation}

For M-ons, using Theorem (\ref{theorem111}), the Hamiltonian is noninvariant under permutations.
The permutation phase for M-ons can not be constructed in the scheme.

The Hilbert subspace describing the M-ons is the space with nonzero
coefficients.

\subsubsection{The maximum occupation number}

One can find, by using Theorem (\ref{theorem2}), 
that the maximum occupation number is $q=\lambda_{1}$. 
In the case of M-ons, 
the only counted microstate is that there are
$\lambda_{1}$ particles occupying a quantum states, $\lambda_{2}$ particles
occupying another quantum state, and so on.

\subsection{generalized quantum statistics corresponding to the power sum symmetric
function: P-ons}

In this section, we propose a kind of generalized quantum statistics corresponding to the
power sum symmetric function $p_{\left(  \lambda\right)}$. For the sake
of convenience, we denote the particle P-ons.

\subsubsection{The canonical partition function}

The canonical partition function of P-ons labeled by the integer
partition $\left(  \lambda\right)$ is
\begin{equation}
Z_{\left(  \lambda\right)  }^{P}\left(  \beta,N\right)  =p_{\left(
\lambda\right)  }\left(  e^{-\beta\varepsilon_{1}},e^{-\beta\varepsilon_{2}%
},...\right)  .\label{Pons}
\end{equation}
Interestingly, one can verify that the P-ons recovers the distinguishable particle when $\left(
\lambda\right)  =\left(  1^{N}\right)  $.

\subsubsection{Discussions on the permutation phase of the wave
function}

Expressing the canonical partition function, Eq. (\ref{Pons}), in terms of the S-function gives%
\begin{align}
Z_{\left(  \lambda\right)  _{K}}^{P}\left(  \beta,N\right)   &  =p_{\left(
\lambda\right)  _{K}}\left(  e^{-\beta\varepsilon_{1}},e^{-\beta
\varepsilon_{2}},...\right)  \nonumber\\
&  =\sum_{I=1}^{P\left(  N\right)  }\chi_{K}^{I}s_{\left(  \lambda\right)
_{I}}\left(  x_{1},x_{2},...\right).  \label{P1}%
\end{align}
The coefficient in Eq. (\ref{P1}) is not nonnegative. For example, the
explicit expressions of the canonical partition function for P-ons labeled by
$\left(  \lambda\right)  =\left(  5\right)  $, $\left(  4,1\right)  $, ...,
are %

\begin{equation}
Z_{\left(  5\right)  }^{p}\left(  5,\beta\right)  =s_{\left(  5\right)
}-s_{\left(  4,1\right)  }+s_{\left(  3,1^{2}\right)  }-s_{\left(
2,1^{3}\right)  }+s_{\left(  1^{5}\right)  },\label{P2}%
\end{equation}%
\begin{equation}
Z_{\left(  4,1\right)  }^{p}\left(  5,\beta\right)  =s_{\left(  5\right)
}-s_{\left(  3,2\right)  }+s_{\left(  2^{2},1\right)  }-s_{\left(
1^{5}\right)  },
\end{equation}%
\begin{equation}
Z_{\left(  3,2\right)  }^{p}\left(  5,\beta\right)  =s_{\left(  5\right)
}-s_{\left(  4,1\right)  }+s_{\left(  3,2\right)  }-s_{\left(  2^{2},1\right)
}+s_{\left(  2,1^{3}\right)  }-s_{\left(  1^{5}\right)  },
\end{equation}%
\begin{equation}
Z_{\left(  3,1^{2}\right)  }^{p}\left(  5,\beta\right)  =s_{\left(  5\right)
}+s_{\left(  4,1\right)  }-s_{\left(  3,2\right)  }-s_{\left(  2^{2},1\right)
}+s_{\left(  2,1^{3}\right)  }+s_{\left(  1^{5}\right)  },
\end{equation}%
\begin{equation}
Z_{\left(  2^{2},1\right)  }^{p}\left(  5,\beta\right)  =s_{\left(  5\right)
}+s_{\left(  3,2\right)  }-2s_{\left(  3,1^{2}\right)  }+s_{\left(
2^{2},1\right)  }+s_{\left(  1^{5}\right)  },
\end{equation}%
\begin{equation}
Z_{\left(  2,1^{3}\right)  }^{p}\left(  5,\beta\right)  =s_{\left(  5\right)
}+s_{\left(  4,1\right)  }-s_{\left(  3,2\right)  }-s_{\left(  2^{2},1\right)
}+s_{\left(  2,1^{3}\right)  }+s_{\left(  1^{5}\right)  },
\end{equation}%
\begin{equation}
Z_{\left(  1^{5}\right)  }^{p}\left(  5,\beta\right)  =s_{\left(  5\right)
}+4s_{\left(  4,1\right)  }+5s_{\left(  3,2\right)  }+6s_{\left(
3,1^{2}\right)  }+5s_{\left(  2^{2},1\right)  }+4s_{\left(  2,1^{3}\right)
}+s_{\left(  1^{5}\right)  }.\label{P3}%
\end{equation}
Thus, the Hamiltonian, using Theorem (\ref{theorem111}), is noninvariant under permutations.
The permutation phase of P-ons can not be constructed in the scheme.

The Hilbert subspace describing the P-ons is the space with nonzero coefficients.

\subsubsection{The maximum occupation number}

Expressing the canonical partition function, Eq. (\ref{Pons}), in terms of the monomical
symmetric function gives%
\begin{align}
Z_{\left(  \lambda\right)  _{K}}^{P}\left(  \beta,N\right)   &  =p_{\left(
\lambda\right)  _{K}}\left(  e^{-\beta\varepsilon_{1}},e^{-\beta
\varepsilon_{2}},...\right)  \nonumber\\
&  =\sum_{j=1}^{P\left(  N\right)  }\left(  \sum_{I=1}^{P\left(  N\right)
}\chi_{K}^{I}k_{I}^{J}\right)  m_{\left(  \lambda\right)  _{J}}\left(
e^{-\beta\varepsilon_{1}},e^{-\beta\varepsilon_{2}},...\right)  . \label{P4}%
\end{align}
The maximum occupation number can be obtained from the non-zero coefficient in
Eq. (\ref{P4}). Based on the property of the simple characteristics of the
$S_{N}$ and the Kostka number $k_{I}^{J}$, Eq. (\ref{P4}) can be written as
\begin{equation}
Z_{\left(  \lambda\right)  _{K}}^{P}\left(  \beta,N\right)  =m_{\left(
N\right)  }\left(  e^{-\beta\varepsilon_{1}},e^{-\beta\varepsilon_{2}%
},...\right)  +\ldots.
\end{equation}
For example, the explicit expression of the canonical partition function for
P-ons labeled by $\left(  \lambda\right)  =\left(  5\right)  $, $\left(
4,1\right)  $, ..., are%
\begin{equation}
Z_{\left(  5\right)  }^{p}\left(  5,\beta\right)  =m_{\left(  5\right)  },
\label{P5}%
\end{equation}%
\begin{equation}
Z_{\left(  4,1\right)  }^{p}\left(  5,\beta\right)  =m_{\left(  5\right)
}+m_{\left(  4,1\right)  },%
\end{equation}%
\begin{equation}
Z_{\left(  3,2\right)  }^{p}\left(  5,\beta\right)  =m_{\left(  5\right)
}+m_{\left(  3,2\right)  },%
\end{equation}%
\begin{equation}
Z_{\left(  3,1^{2}\right)  }^{p}\left(  5,\beta\right)  =m_{\left(  5\right)
}+2m_{\left(  4,1\right)  }+m_{\left(  3,2\right)  }+2m_{\left(
2^{2},1\right)  },%
\end{equation}%
\begin{equation}
Z_{\left(  2^{2},1\right)  }^{p}\left(  5,\beta\right)  =m_{\left(  5\right)
}+m_{\left(  4,1\right)  }+2m_{\left(  3,2\right)  }+2s_{\left(
2^{2},1\right)  },%
\end{equation}%
\begin{equation}
Z_{\left(  2,1^{3}\right)  }^{p}\left(  5,\beta\right)  =m_{\left(  5\right)
}+2m_{\left(  4,1\right)  }+m_{\left(  3,2\right)  }+2m_{\left(
3,1^{2}\right)  },%
\end{equation}%
\begin{align}
Z_{\left(  1^{5}\right)  }^{p}\left(  5,\beta\right)   &  =m_{\left(
5\right)  }+5m_{\left(  4,1\right)  }+10m_{\left(  3,2\right)  }+20m_{\left(
3,1^{2}\right)  }\nonumber\\
&  +30m_{\left(  2^{2},1\right)  }+60m_{\left(  2,1^{3}\right)  }%
+120m_{\left(  1^{5}\right)  }. \label{P7}%
\end{align}
One can find, by using Theorem (\ref{theorem2}), that for P-ons, 
there is no limitation on the maximum occupation number.

\section{Conclusions}

In this paper, we give the unified
framework to describe various kinds of generalized quantum statistics. 
We provide a general formula of canonical partition functions of ideal $N$-particle gases obeying 
various kinds of generalized quantum statistics. We reveal the connection between the permutation 
phase of the wave function and the maximum occupation number, through constructing a method of 
obtaining the permutation phase and the maximum occupation number from the canonical partition function. 

We show that for generalized quantum statistics, the permutation 
phase of wave functions should be generalized to a matrix, rather than a number. 
The permutation phase of Bose or Fermi wave function, $1$ or $-1$, is regarded as $1\times1$ matrices, 
as special cases of generalized quantum statistics. We suggest a method to distinguish
generalized quantum statistics with the Hamiltonian noninvariant under permutations, 
in which the permutation phase can not be constructed in the scheme. 

It is commonly accepted that different kinds of statistics are distinguished by the maximum 
number. We show that the maximum occupation number is not sufficient to distinguish different kinds 
of generalized quantum statistics. 

As examples, we describe various kinds of statistics in a unified framework,
including parastatistics, Gentile statistics, Gentileonic
statistics, and the immannons. 
The canonical partition function, the maximum occupation number, 
the permutation phase of the wave function, 
and the Hilbert subspace are given. Especially, we propose three new kinds of generalized 
statistics which seem to be the missing pieces in the puzzle.

The mathematical basis of the scheme are the 
mathematical theory of the invariant matrix, the Schur-Weyl duality, the symmetric function, and 
the representation theory of the permutation group and the unitary group. The result in this paper,
together with our previous works
\cite{zhou2018canonical,zhou2018statistical}, 
builds a bridge between the quantum statistical mechanics and such mathematical theories.
This enables one to use the fruitful result in such theories to solve the
problem in quantum statistical mechanics.

\section{Acknowledgments}
We are very indebted to Dr G Zeitrauman for his encouragement. This work is supported
in part by NSF of China under Grant No. 11575125 and No. 11675119.

\section{Appendix}

\subsection{The Kostka number}

$N=3$\textit{.} The Kostka number $k_{K}^{J}$ is $k_{\left(
3\right)  }^{\left(  3\right)  }=k_{1}^{1}=1$, $k_{\left(  2,1\right)
}^{\left(  3\right)  }=k_{2}^{1}=0$, $k_{\left(  1^{3}\right)  }^{\left(
3\right)  }=k_{3}^{1}=0$, $k_{\left(  3\right)  }^{\left(  2,1\right)  }%
=k_{1}^{2}=1$, $k_{\left(  2,1\right)  }^{\left(  2,1\right)  }=k_{2}^{2}=1$,
$k_{\left(  1^{3}\right)  }^{\left(  2,1\right)  }=k_{3}^{2}=0$, $k_{\left(
3\right)  }^{\left(  1^{3}\right)  }=k_{1}^{3}=1$, $k_{\left(  2,1\right)
}^{\left(  1^{3}\right)  }=k_{2}^{3}=2$, and $k_{\left(  1^{3}\right)
}^{\left(  1^{3}\right)  }=k_{3}^{3}=1$
\cite{littlewood1977theory,macdonald1998symmetric}. For clarity , we
rewrite the Kostka number in a matrix form:%
\begin{equation}
k=\left(
\begin{array}
[c]{ccc}%
1 & 0 & 0\\
1 & 1 & 0\\
1 & 2 & 1
\end{array}
\right)  , \label{a1}%
\end{equation}
where we take the upper index of $k_{K}^{J}$ as a row index and the lower
index as a column index.

$N=4$. The Kostka number is
\cite{littlewood1977theory,macdonald1998symmetric}%

\begin{equation}
k=\left(
\begin{array}
[c]{ccccc}%
1 & 0 & 0 & 0 & 0\\
1 & 1 & 0 & 0 & 0\\
1 & 1 & 1 & 0 & 0\\
1 & 2 & 1 & 1 & 0\\
1 & 3 & 2 & 3 & 1
\end{array}
\right)  . \label{a5}%
\end{equation}

$N=5$. The Kostka number is
\cite{littlewood1977theory,macdonald1998symmetric}%
\begin{equation}
k=\left(
\begin{array}
[c]{ccccccc}%
1 & 0 & 0 & 0 & 0 & 0 & 0\\
1 & 1 & 0 & 0 & 0 & 0 & 0\\
1 & 1 & 1 & 0 & 0 & 0 & 0\\
1 & 2 & 1 & 1 & 0 & 0 & 0\\
1 & 2 & 2 & 1 & 1 & 0 & 0\\
1 & 3 & 3 & 3 & 2 & 1 & 0\\
1 & 4 & 5 & 6 & 5 & 4 & 1
\end{array}
\right)  . \label{a10}%
\end{equation}

\subsection{The simple characteristic of $S_{N}$}

$N=3$\textit{.} The simple characteristic $\chi_{K}^{J}$ is
$\chi_{\left(  3\right)  }^{\left(  3\right)  }=\chi_{1}^{1}=1$,
$\chi_{\left(  2,1\right)  }^{\left(  3\right)  }=\chi_{2}^{1}=1$,
$\chi_{\left(  1^{3}\right)  }^{\left(  3\right)  }=\chi_{3}^{1}=1$,
$\chi_{\left(  3\right)  }^{\left(  2,1\right)  }=\chi_{1}^{2}=2$,
$\chi_{\left(  2,1\right)  }^{\left(  2,1\right)  }=\chi_{2}^{2}=0$,
$\chi_{\left(  1^{3}\right)  }^{\left(  2,1\right)  }=\chi_{3}^{2}=-1$,
$\chi_{\left(  3\right)  }^{\left(  1^{3}\right)  }=\chi_{1}^{3}=1$,
$\chi_{\left(  2,1\right)  }^{\left(  1^{3}\right)  }=\chi_{2}^{3}=-1$, and
$\chi_{\left(  1^{3}\right)  }^{\left(  1^{3}\right)  }=\chi_{3}^{3}=1$
\cite{hamermesh1962group}. For clarity , we can rewrite the simple
characteristic in a matrix form:%

\begin{equation}
\chi=\left(
\begin{array}
[c]{ccc}%
1 & 1 & 1\\
2 & 0 & -1\\
1 & -1 & 1
\end{array}
\right)  . \label{2c1}%
\end{equation}

$N=4$. The $\chi$ is \cite{hamermesh1962group}
\begin{equation}
\chi=\left(
\begin{array}
[c]{ccccc}%
1 & 1 & 1 & 1 & 1\\
3 & 0 & -1 & 1 & -1\\
2 & -1 & 2 & 0 & 0\\
3 & 0 & -1 & -1 & 1\\
1 & 1 & 1 & -1 & -1
\end{array}
\right)  . \label{2c2}%
\end{equation}

$N=5$. The $\chi$ is \cite{hamermesh1962group}%
\begin{equation}
\chi=\left(
\begin{array}
[c]{ccccccc}%
1 & 1 & 1 & 1 & 1 & 1 & 1\\
-1 & 0 & -1 & 1 & 0 & 1 & 4\\
0 & -1 & 1 & -1 & 1 & -1 & 5\\
1 & 0 & 0 & 0 & -2 & 0 & 6\\
0 & 1 & -1 & -1 & 1 & -1 & 5\\
-1 & 0 & 1 & 1 & 0 & 1 & 4\\
1 & -1 & -1 & 1 & 1 & 1 & 1
\end{array}
\right)  . \label{2c3}%
\end{equation}

\subsection{The calculation of the maximum occupation number of the
parastatistics: an example of calculation the coefficient in the canonical partition function}

The canonical partition function of an $N$-particle gas under various kinds of generalized quantum statistics
can be expressed in terms of symmetric functions such as the S-function, the
monomical symmetric function, and so on. The procedure of 
the calculation of the coefficient in the canonical partition function
involves the transformation of the coefficient by the
matrix of Kostka number or the simple characteristic.
In this appendix, we
give details of the calculation of the coefficient for parastatistics
as an example.

$N=3$, $q=1$, for para-Bose statistics, the coefficient in Eq. (\ref{para3}) is%

\begin{equation}
\left(
\begin{array}
[c]{ccc}%
1 & 0 & 0\\
1 & 1 & 0\\
1 & 2 & 1
\end{array}
\right)  \left(
\begin{array}
[c]{c}%
1\\
0\\
0
\end{array}
\right)  =\left(
\begin{array}
[c]{c}%
1\\
1\\
1
\end{array}
\right)  ,
\end{equation}
that is,
\begin{align*}
Z_{1}^{PB}\left(  \beta,3\right)   &  =m_{\left(  3\right)  }  +m_{\left(
2,1\right)  } +m_{\left(  1^{3}\right)  };
\end{align*}
for para-Fermi statistics, the coefficient in Eq. (\ref{para4}) is%
\begin{equation}
\left(
\begin{array}
[c]{ccc}%
1 & 0 & 0\\
1 & 1 & 0\\
1 & 2 & 1
\end{array}
\right)  \left(
\begin{array}
[c]{c}%
0\\
0\\
1
\end{array}
\right)  =\left(
\begin{array}
[c]{c}%
0\\
0\\
1
\end{array}
\right)  ,
\end{equation}%
that is,
\[
Z_{1}^{PB}\left(  \beta,3\right)  =m_{\left(  1^{3}\right)  }.
\]

In the following, the corresponding
expression of the canonical partition 
function is already given in Eqs. (\ref{para001})-(\ref{para02}), thus, we
only give the coefficient. $q=2$, for para-Bose statistics, the coefficient in Eq. (\ref{para3}) is%
\begin{equation}
\left(
\begin{array}
[c]{ccc}%
1 & 0 & 0\\
1 & 1 & 0\\
1 & 2 & 1
\end{array}
\right)  \left(
\begin{array}
[c]{c}%
1\\
1\\
0
\end{array}
\right)  =\left(
\begin{array}
[c]{c}%
1\\
2\\
3
\end{array}
\right)  ;
\end{equation} for para-Fermi statistics, the coefficient in Eq. (\ref{para4}) is%
\begin{equation}
\left(
\begin{array}
[c]{ccc}%
1 & 0 & 0\\
1 & 1 & 0\\
1 & 2 & 1
\end{array}
\right)  \left(
\begin{array}
[c]{c}%
0\\
1\\
1
\end{array}
\right)  =\left(
\begin{array}
[c]{c}%
0\\
1\\
3
\end{array}
\right)  .
\end{equation}$q=3$, for para-Bose statistics, the coefficient in Eq. (\ref{para3}) is%
\begin{equation}
\left(
\begin{array}
[c]{ccc}%
1 & 0 & 0\\
1 & 1 & 0\\
1 & 2 & 1
\end{array}
\right)  \left(
\begin{array}
[c]{c}%
1\\
1\\
1
\end{array}
\right)  =\left(
\begin{array}
[c]{c}%
1\\
2\\
3
\end{array}
\right)  ;
\end{equation} for para-Fermi statistics, the coefficient in Eq. (\ref{para4}) is%
\begin{equation}
\left(
\begin{array}
[c]{ccc}%
1 & 0 & 0\\
1 & 1 & 0\\
1 & 2 & 1
\end{array}
\right)  \left(
\begin{array}
[c]{c}%
1\\
1\\
1
\end{array}
\right)  =\left(
\begin{array}
[c]{c}%
1\\
2\\
3
\end{array}
\right)  .
\end{equation}%

$N=4$, $q=1$, for para-Bose statistics, the coefficient in Eq. (\ref{para3}) is%

\begin{equation}
\left(
\begin{array}
[c]{ccccc}%
1 & 0 & 0 & 0 & 0\\
1 & 1 & 0 & 0 & 0\\
1 & 1 & 1 & 0 & 0\\
1 & 2 & 1 & 1 & 0\\
1 & 3 & 2 & 3 & 1
\end{array}
\right)  \left(
\begin{array}
[c]{c}%
1\\
0\\
0\\
0\\
0
\end{array}
\right)  =\left(
\begin{array}
[c]{c}%
1\\
1\\
1\\
1\\
1
\end{array}
\right)  ;
\end{equation} for para-Fermi statistics, the coefficient in Eq. (\ref{para4}) is%
\begin{equation}
\left(
\begin{array}
[c]{ccccc}%
1 & 0 & 0 & 0 & 0\\
1 & 1 & 0 & 0 & 0\\
1 & 1 & 1 & 0 & 0\\
1 & 2 & 1 & 1 & 0\\
1 & 3 & 2 & 3 & 1
\end{array}
\right)  \left(
\begin{array}
[c]{c}%
0\\
0\\
0\\
0\\
1
\end{array}
\right)  =\left(
\begin{array}
[c]{c}%
0\\
0\\
0\\
0\\
1
\end{array}
\right)  .
\end{equation} $q=2$, for para-Bose statistics, the coefficient in Eq. (\ref{para3}) is%

\begin{equation}
\left(
\begin{array}
[c]{ccccc}%
1 & 0 & 0 & 0 & 0\\
1 & 1 & 0 & 0 & 0\\
1 & 1 & 1 & 0 & 0\\
1 & 2 & 1 & 1 & 0\\
1 & 3 & 2 & 3 & 1
\end{array}
\right)  \left(
\begin{array}
[c]{c}%
1\\
1\\
1\\
0\\
0
\end{array}
\right)  =\left(
\begin{array}
[c]{c}%
1\\
2\\
3\\
4\\
6
\end{array}
\right)  ;
\end{equation} for para-Fermi statistics, the coefficient in Eq. (\ref{para4}) is%
\begin{equation}
\left(
\begin{array}
[c]{ccccc}%
1 & 0 & 0 & 0 & 0\\
1 & 1 & 0 & 0 & 0\\
1 & 1 & 1 & 0 & 0\\
1 & 2 & 1 & 1 & 0\\
1 & 3 & 2 & 3 & 1
\end{array}
\right)  \left(
\begin{array}
[c]{c}%
0\\
0\\
1\\
1\\
1
\end{array}
\right)  =\left(
\begin{array}
[c]{c}%
0\\
0\\
1\\
2\\
6
\end{array}
\right)  .
\end{equation} $q=3$, for para-Bose statistics, the coefficient in Eq. (\ref{para3}) is%

\begin{equation}
\left(
\begin{array}
[c]{ccccc}%
1 & 0 & 0 & 0 & 0\\
1 & 1 & 0 & 0 & 0\\
1 & 1 & 1 & 0 & 0\\
1 & 2 & 1 & 1 & 0\\
1 & 3 & 2 & 3 & 1
\end{array}
\right)  \left(
\begin{array}
[c]{c}%
1\\
1\\
1\\
1\\
0
\end{array}
\right)  =\left(
\begin{array}
[c]{c}%
1\\
2\\
3\\
4\\
9
\end{array}
\right)  ;
\end{equation} for para-Fermi statistics, the coefficient in Eq. (\ref{para4}) is%
\begin{equation}
\left(
\begin{array}
[c]{ccccc}%
1 & 0 & 0 & 0 & 0\\
1 & 1 & 0 & 0 & 0\\
1 & 1 & 1 & 0 & 0\\
1 & 2 & 1 & 1 & 0\\
1 & 3 & 2 & 3 & 1
\end{array}
\right)  \left(
\begin{array}
[c]{c}%
0\\
1\\
1\\
1\\
1
\end{array}
\right)  =\left(
\begin{array}
[c]{c}%
0\\
1\\
2\\
4\\
9
\end{array}
\right)  .
\end{equation} $q=4$, for para-Bose statistics, the coefficient in Eq. (\ref{para3}) is%
\begin{equation}
\left(
\begin{array}
[c]{ccccc}%
1 & 0 & 0 & 0 & 0\\
1 & 1 & 0 & 0 & 0\\
1 & 1 & 1 & 0 & 0\\
1 & 2 & 1 & 1 & 0\\
1 & 3 & 2 & 3 & 1
\end{array}
\right)  \left(
\begin{array}
[c]{c}%
1\\
1\\
1\\
1\\
1
\end{array}
\right)  =\left(
\begin{array}
[c]{c}%
1\\
2\\
3\\
5\\
10
\end{array}
\right)  ;
\end{equation} for para-Fermi statistics, the coefficient in Eq. (\ref{para4}) is%
\begin{equation}
\left(
\begin{array}
[c]{ccccc}%
1 & 0 & 0 & 0 & 0\\
1 & 1 & 0 & 0 & 0\\
1 & 1 & 1 & 0 & 0\\
1 & 2 & 1 & 1 & 0\\
1 & 3 & 2 & 3 & 1
\end{array}
\right)  \left(
\begin{array}
[c]{c}%
1\\
1\\
1\\
1\\
1
\end{array}
\right)  =\left(
\begin{array}
[c]{c}%
1\\
2\\
3\\
5\\
10
\end{array}
\right)  .
\end{equation}%

$N=5$, $q=1$, for para-Bose statistics, the coefficient in Eq. (\ref{para3}) is%

\[
\left(
\begin{array}
[c]{ccccccc}%
1 & 0 & 0 & 0 & 0 & 0 & 0\\
1 & 1 & 0 & 0 & 0 & 0 & 0\\
1 & 1 & 1 & 0 & 0 & 0 & 0\\
1 & 2 & 1 & 1 & 0 & 0 & 0\\
1 & 2 & 2 & 1 & 1 & 0 & 0\\
1 & 3 & 3 & 3 & 2 & 1 & 0\\
1 & 4 & 5 & 6 & 5 & 4 & 1
\end{array}
\right)  \left(
\begin{array}
[c]{c}%
1\\
0\\
0\\
0\\
0\\
0\\
0
\end{array}
\right)  =\left(
\begin{array}
[c]{c}%
1\\
1\\
1\\
1\\
1\\
1\\
1
\end{array}
\right);
\] for para-Fermi statistics, the coefficient in Eq. (\ref{para4}) is%
\[
\left(
\begin{array}
[c]{ccccccc}%
1 & 0 & 0 & 0 & 0 & 0 & 0\\
1 & 1 & 0 & 0 & 0 & 0 & 0\\
1 & 1 & 1 & 0 & 0 & 0 & 0\\
1 & 2 & 1 & 1 & 0 & 0 & 0\\
1 & 2 & 2 & 1 & 1 & 0 & 0\\
1 & 3 & 3 & 3 & 2 & 1 & 0\\
1 & 4 & 5 & 6 & 5 & 4 & 1
\end{array}
\right)  \left(
\begin{array}
[c]{c}%
0\\
0\\
0\\
0\\
0\\
0\\
1
\end{array}
\right)  =\left(
\begin{array}
[c]{c}%
0\\
0\\
0\\
0\\
0\\
0\\
1
\end{array}
\right).
\] $q=2$, for para-Bose statistics, the coefficient in Eq. (\ref{para3}) is%

\[
\left(
\begin{array}
[c]{ccccccc}%
1 & 0 & 0 & 0 & 0 & 0 & 0\\
1 & 1 & 0 & 0 & 0 & 0 & 0\\
1 & 1 & 1 & 0 & 0 & 0 & 0\\
1 & 2 & 1 & 1 & 0 & 0 & 0\\
1 & 2 & 2 & 1 & 1 & 0 & 0\\
1 & 3 & 3 & 3 & 2 & 1 & 0\\
1 & 4 & 5 & 6 & 5 & 4 & 1
\end{array}
\right)  \left(
\begin{array}
[c]{c}%
1\\
1\\
1\\
0\\
0\\
0\\
0
\end{array}
\right)  =\left(
\begin{array}
[c]{c}%
1\\
2\\
3\\
4\\
5\\
7\\
10
\end{array}
\right);
\] for para-Fermi statistics, the coefficient in Eq. (\ref{para4}) is%
\[
\left(
\begin{array}
[c]{ccccccc}%
1 & 0 & 0 & 0 & 0 & 0 & 0\\
1 & 1 & 0 & 0 & 0 & 0 & 0\\
1 & 1 & 1 & 0 & 0 & 0 & 0\\
1 & 2 & 1 & 1 & 0 & 0 & 0\\
1 & 2 & 2 & 1 & 1 & 0 & 0\\
1 & 3 & 3 & 3 & 2 & 1 & 0\\
1 & 4 & 5 & 6 & 5 & 4 & 1
\end{array}
\right)  \left(
\begin{array}
[c]{c}%
0\\
0\\
0\\
0\\
1\\
1\\
1
\end{array}
\right)  =\left(
\begin{array}
[c]{c}%
0\\
0\\
0\\
0\\
1\\
3\\
10
\end{array}
\right).
\] $q=3$, for para-Bose statistics, the coefficient in Eq. (\ref{para3}) is%

\[
\left(
\begin{array}
[c]{ccccccc}%
1 & 0 & 0 & 0 & 0 & 0 & 0\\
1 & 1 & 0 & 0 & 0 & 0 & 0\\
1 & 1 & 1 & 0 & 0 & 0 & 0\\
1 & 2 & 1 & 1 & 0 & 0 & 0\\
1 & 2 & 2 & 1 & 1 & 0 & 0\\
1 & 3 & 3 & 3 & 2 & 1 & 0\\
1 & 4 & 5 & 6 & 5 & 4 & 1
\end{array}
\right)  \left(
\begin{array}
[c]{c}%
1\\
1\\
1\\
1\\
1\\
0\\
0
\end{array}
\right)  =\left(
\begin{array}
[c]{c}%
1\\
2\\
3\\
5\\
7\\
12\\
21
\end{array}
\right);
\] for para-Fermi statistics, the coefficient in Eq. (\ref{para4}) is%
\[
\left(
\begin{array}
[c]{ccccccc}%
1 & 0 & 0 & 0 & 0 & 0 & 0\\
1 & 1 & 0 & 0 & 0 & 0 & 0\\
1 & 1 & 1 & 0 & 0 & 0 & 0\\
1 & 2 & 1 & 1 & 0 & 0 & 0\\
1 & 2 & 2 & 1 & 1 & 0 & 0\\
1 & 3 & 3 & 3 & 2 & 1 & 0\\
1 & 4 & 5 & 6 & 5 & 4 & 1
\end{array}
\right)  \left(
\begin{array}
[c]{c}%
0\\
0\\
1\\
1\\
1\\
1\\
1
\end{array}
\right)  =\left(
\begin{array}
[c]{c}%
0\\
0\\
1\\
2\\
4\\
9\\
21
\end{array}
\right).
\] $q=4$, for para-Bose statistics, the coefficient in Eq. (\ref{para3}) is%

\[
\left(
\begin{array}
[c]{ccccccc}%
1 & 0 & 0 & 0 & 0 & 0 & 0\\
1 & 1 & 0 & 0 & 0 & 0 & 0\\
1 & 1 & 1 & 0 & 0 & 0 & 0\\
1 & 2 & 1 & 1 & 0 & 0 & 0\\
1 & 2 & 2 & 1 & 1 & 0 & 0\\
1 & 3 & 3 & 3 & 2 & 1 & 0\\
1 & 4 & 5 & 6 & 5 & 4 & 1
\end{array}
\right)  \left(
\begin{array}
[c]{c}%
1\\
1\\
1\\
1\\
1\\
1\\
0
\end{array}
\right)  =\left(
\begin{array}
[c]{c}%
1\\
2\\
3\\
5\\
7\\
13\\
25
\end{array}
\right);
\] for para-Fermi statistics, the coefficient in Eq. (\ref{para4}) is%
\[
\left(
\begin{array}
[c]{ccccccc}%
1 & 0 & 0 & 0 & 0 & 0 & 0\\
1 & 1 & 0 & 0 & 0 & 0 & 0\\
1 & 1 & 1 & 0 & 0 & 0 & 0\\
1 & 2 & 1 & 1 & 0 & 0 & 0\\
1 & 2 & 2 & 1 & 1 & 0 & 0\\
1 & 3 & 3 & 3 & 2 & 1 & 0\\
1 & 4 & 5 & 6 & 5 & 4 & 1
\end{array}
\right)  \left(
\begin{array}
[c]{c}%
0\\
1\\
1\\
1\\
1\\
1\\
1
\end{array}
\right)  =\left(
\begin{array}
[c]{c}%
0\\
1\\
2\\
4\\
6\\
12\\
25
\end{array}
\right).
\] $q=5$, for para-Bose statistics, the coefficient in Eq. (\ref{para3}) is%

\[
\left(
\begin{array}
[c]{ccccccc}%
1 & 0 & 0 & 0 & 0 & 0 & 0\\
1 & 1 & 0 & 0 & 0 & 0 & 0\\
1 & 1 & 1 & 0 & 0 & 0 & 0\\
1 & 2 & 1 & 1 & 0 & 0 & 0\\
1 & 2 & 2 & 1 & 1 & 0 & 0\\
1 & 3 & 3 & 3 & 2 & 1 & 0\\
1 & 4 & 5 & 6 & 5 & 4 & 1
\end{array}
\right)  \left(
\begin{array}
[c]{c}%
1\\
1\\
1\\
1\\
1\\
1\\
1
\end{array}
\right)  =\left(
\begin{array}
[c]{c}%
1\\
2\\
3\\
5\\
7\\
13\\
26
\end{array}
\right);
\] for para-Fermi statistics, the coefficient in Eq. (\ref{para4}) is%
\[
\left(
\begin{array}
[c]{ccccccc}%
1 & 0 & 0 & 0 & 0 & 0 & 0\\
1 & 1 & 0 & 0 & 0 & 0 & 0\\
1 & 1 & 1 & 0 & 0 & 0 & 0\\
1 & 2 & 1 & 1 & 0 & 0 & 0\\
1 & 2 & 2 & 1 & 1 & 0 & 0\\
1 & 3 & 3 & 3 & 2 & 1 & 0\\
1 & 4 & 5 & 6 & 5 & 4 & 1
\end{array}
\right)  \left(
\begin{array}
[c]{c}%
1\\
1\\
1\\
1\\
1\\
1\\
1
\end{array}
\right)  =\left(
\begin{array}
[c]{c}%
1\\
2\\
3\\
5\\
7\\
13\\
26
\end{array}
\right).
\]%


\acknowledgments

We are very indebted to Dr G. Zeitrauman for his encouragement. This work is supported in part by NSF of China under Grant No.
11575125 and No.  11675119.








\providecommand{\href}[2]{#2}\begingroup\raggedright\endgroup



\begin{thebibliography}{10}

\bibitem{reichl2009modern}
L.~Reichl, {\em A Modern Course in Statistical Physics}.
\newblock Physics textbook. Wiley, 2009.

\bibitem{pathria2011statistical}
R.~Pathria, {\em Statistical Mechanics}.
\newblock Elsevier Science, 2011.

\bibitem{leinaas1977theory}
J.~M. Leinaas and J.~Myrheim, {\it On the theory of identical particles},  {\em
  Il Nuovo Cimento B (1971-1996)} {\bf 37} (1977), no.~1 1--23.

\bibitem{khare2005fractional}
A.~Khare, {\em Fractional statistics and quantum theory}.
\newblock World Scientific, 2005.

\bibitem{haldane1991fractional}
F.~D.~M. Haldane, {\it 'fractional statistics'in arbitrary dimensions: A
  generalization of the pauli principle},  {\em Physical review letters} {\bf
  67} (1991), no.~8 937.

\bibitem{dai2009intermediate}
W.-S. Dai and M.~Xie, {\it Intermediate-statistics spin waves},  {\em Journal
  of Statistical Mechanics: Theory and Experiment} {\bf 2009} (2009), no.~04
  P04021.

\bibitem{zhou2018calculating}
C.-C. Zhou and W.-S. Dai, {\it Calculating eigenvalues of many-body systems
  from partition functions},  {\em Journal of Statistical Mechanics: Theory and
  Experiment} {\bf 2018} (2018), no.~8 083103.

\bibitem{zhou2018canonical}
C.-C. Zhou and W.-S. Dai, {\it Canonical partition functions: ideal quantum
  gases, interacting classical gases, and interacting quantum gases},  {\em
  Journal of Statistical Mechanics: Theory and Experiment} {\bf 2018} (2018),
  no.~2 023105.

\bibitem{littlewood1977theory}
D.~E. Littlewood, {\em The theory of group characters and matrix
  representations of groups}, vol.~357.
\newblock American Mathematical Soc., 1977.

\bibitem{meijer2017schur}
R.~Meijer, {\it Schur-weyl duality},  {B.S.} thesis, 2017.

\bibitem{macdonald1998symmetric}
I.~G. Macdonald, {\em Symmetric functions and Hall polynomials}.
\newblock Oxford university press, 1998.

\bibitem{green1953generalized}
H.~S. Green, {\it A generalized method of field quantization},  {\em Physical
  Review} {\bf 90} (1953), no.~2 270.

\bibitem{ohnuki1982quantum}
Y.~Ohnuki and S.~Kamefuchi, {\it Quantum field theory and parastatistics}, .

\bibitem{gentile1940itosservazioni}
G.~Gentile~j, {\it Itosservazioni sopra le statistiche intermedie},  {\em Il
  Nuovo Cimento (1924-1942)} {\bf 17} (1940) 493--497.

\bibitem{dai2004gentile}
W.-S. Dai and M.~Xie, {\it Gentile statistics with a large maximum occupation
  number},  {\em Annals of Physics} {\bf 309} (2004), no.~2 295--305.

\bibitem{cattani1984general}
M.~Cattani and N.~C. Fernandes, {\it General statistics, second quantization
  and quarks},  {\em Il Nuovo Cimento A (1965-1970)} {\bf 79} (1984), no.~1
  107.

\bibitem{tichy2017extending}
M.~C. Tichy and K.~M{\o}lmer, {\it Extending exchange symmetry beyond bosons
  and fermions},  {\em Physical Review A} {\bf 96} (2017), no.~2 022119.

\bibitem{iachello2006lie}
F.~Iachello, {\em Lie algebras and applications}, vol.~12.
\newblock Springer, 2006.

\bibitem{hamermesh1962group}
M.~Hamermesh, {\em Group theory and its application to physical problems}.
\newblock Courier Corporation, 2012.

\bibitem{cattani2009intermediate}
M.~Cattani and J.~M.~F. Bassalo, {\it Intermediate statistics, parastatistics,
  fractionary statistics and gentileonic statistics},  {\em arXiv preprint
  arXiv:0903.4773} (2009).

\bibitem{katsura1970intermediate}
S.~Katsura, K.~Kaminishi, and S.~Inawashiro, {\it Intermediate statistics},
  {\em Journal of Mathematical Physics} {\bf 11} (1970), no.~9 2691--2697.

\bibitem{okayama1952generalization}
T.~Okayama, {\it Generalization of statistics},  {\em Progress of Theoretical
  Physics} {\bf 7} (1952) 517--534.

\bibitem{dai2004representation}
W.-S. Dai and M.~Xie, {\it A representation of angular momentum (su (2))
  algebra},  {\em Physica A: Statistical Mechanics and its Applications} {\bf
  331} (2004), no.~3-4 497--504.

\bibitem{dai2012calculating}
W.-S. Dai and M.~Xie, {\it Calculating statistical distributions from operator
  relations: The statistical distributions of various intermediate statistics},
   {\em Annals of Physics} {\bf 332} (2012) 166--179.

\bibitem{shen2010relation}
Y.~Shen, Q.~Ai, and G.~L. Long, {\it The relation between properties of gentile
  statistics and fractional statistics of anyon},  {\em Physica A: Statistical
  Mechanics and its Applications} {\bf 389} (2010), no.~8 1565--1570.

\bibitem{shen2007intermediate}
Y.~Shen, W.-S. Dai, and M.~Xie, {\it Intermediate-statistics quantum bracket,
  coherent state, oscillator, and representation of angular momentum [su (2)]
  algebra},  {\em Physical Review A} {\bf 75} (2007), no.~4 042111.

\bibitem{wu1994statistical}
Y.-S. Wu, {\it Statistical distribution for generalized ideal gas of
  fractional-statistics particles},  {\em Physical review letters} {\bf 73}
  (1994), no.~7 922.

\bibitem{vilenkin2013representation}
N.~J. Vilenkin and A.~Klimyk, {\em Representation of Lie groups and special
  functions: recent advances}.

\bibitem{andrews1998theory}
G.~E. Andrews, {\em The theory of partitions}.
\newblock No.~2. Cambridge university press, 1998.

\bibitem{goulden1992immanants}
I.~Goulden and D.~Jackson, {\it Immanants, schur functions, and the macmahon
  master theorem},  {\em Proceedings of the American Mathematical Society} {\bf
  115} (1992), no.~3 605--612.

\bibitem{dao2018schur}
Q.~DAO, {\it Schur-weyl duality}, .

\bibitem{zhou2018statistical}
C.-C. Zhou and W.-S. Dai, {\it A statistical mechanical approach to restricted
  integer partition functions},  {\em Journal of Statistical Mechanics: Theory
  and Experiment} {\bf 2018} (2018), no.~5 053111.

\bibitem{maslov2017relationship}
V.~P. Maslov, {\it The relationship between the fermi--dirac distribution and
  statistical distributions in languages},  {\em Mathematical Notes} {\bf 101}
  (2017), no.~3-4 645--659.

\bibitem{chaturvedi1996canonical}
S.~Chaturvedi, {\it Canonical partition functions for parastatistical systems
  of any order},  {\em Physical Review E} {\bf 54} (1996), no.~2 1378.

\bibitem{chaturvedi1996grand}
S.~Chaturvedi and V.~Srinivasan, {\it Grand canonical partition functions for
  multi level para fermi systems of any order},  {\em arXiv preprint
  hep-th/9608150} (1996).

\bibitem{vo1972first}
T.~Vo-Dai, {\em First and second quantization theories of parastatistics}.
\newblock PhD thesis, 1972.

\end{thebibliography}

\end{document}